\documentclass{article}

\usepackage[english]{babel}

\usepackage[letterpaper,top=2cm,bottom=2cm,left=3cm,right=3cm,marginparwidth=1.75cm]{geometry}
\usepackage{authblk}
\usepackage{amsmath}
\usepackage{amsthm}
\usepackage{amssymb}
\usepackage{dsfont}
\usepackage{graphicx}
\usepackage[colorlinks=true, allcolors=blue]{hyperref}
\usepackage{manfnt}

\usepackage{lscape}
\newcommand{\Tr}{{\rm Tr}}
\newcommand{\sq}{\square}
\newcommand{\Sq}{\boxplus}
\newcommand{\A}{\mathcal{A}}
\newcommand{\id}{\mathds{1}}
\newcommand{\ket}[1]{\vert #1 \rangle}
\newcommand{\bra}[1]{\langle #1 \vert}

\def\barray{\begin{eqnarray}}
\def\earray{\end{eqnarray}}
\def\beq{\begin{equation}}
\def\eeq{\end{equation}}
\newtheorem{definition}{Definition}
\newtheorem{proposition}{Proposition}

\usepackage[dvipsnames]{xcolor}
\usepackage{tikz}
\usepackage{tikz-cd}
\usetikzlibrary{decorations.pathmorphing, decorations.markings, calc,3d, positioning, quotes, arrows, arrows.meta, shapes, patterns.meta}
\tikzset{every path/.style={line cap=round},baseline={([yshift=-.5ex]current bounding box.center)},bevel/.style={preaction={draw,white,line width=2pt,line cap=round}},thick bevel/.style={preaction={draw,white,line width=4pt,line cap=round}},disk/.style={circle,draw=black,fill=black},fdisk/.style={circle,draw=black},star8/.style={draw=black,star,star points=8,fill=black},fstar8/.style={draw=black,star,star points=8},disk normal/.style={disk,inner sep=1.45pt},disk large/.style={disk,inner sep=1.75pt},disk small/.style={disk,inner sep=1pt},disk tiny/.style={disk,inner sep=1.45pt},fdisk normal/.style={fdisk,inner sep=1.45pt},fdisk large/.style={fdisk,inner sep=1.75pt},fdisk small/.style={fdisk,inner sep=1pt},fdisk tiny/.style={fdisk,inner sep=1.45pt},star8 normal/.style={star8,star point height=0.25mm,inner sep=1.1pt},star8 small/.style={star8,star point height=0.2mm,inner sep=0.9pt},star8 large/.style={star8,star point height=0.25mm,inner sep=1.45pt},star8 tiny/.style={star8,star point height=0.2mm,inner sep=0.7pt},fstar8 normal/.style={fstar8,star point height=0.25mm,inner sep=1.1pt},fstar8 small/.style={fstar8,star point height=0.2mm,inner sep=0.9pt},fstar8 large/.style={fstar8,star point height=0.25mm,inner sep=1.45pt},fstar8 tiny/.style={fstar8,star point height=0.2mm,inner sep=0.7pt},-mid/.style={decoration={markings,mark=at position 0.5*\pgfdecoratedpathlength+0.6*3pt with \arrow{>[width=3pt]}},postaction={decorate}},mid-/.style={decoration={markings,mark=at position 0.5*\pgfdecoratedpathlength+0.6*3pt with \arrow{<[width=3pt]}},postaction={decorate}}}
\tikzdeclarepattern{name=rising, bounding box={(0,-0.75pt) and (2pt,0.75pt)}, tile size={(2pt, 2pt)}, tile transformation={rotate=45}, code={\draw [line width=1.5pt] (0,0) -- (2pt,0);} }
\tikzdeclarepattern{name=falling, bounding box={(0,-0.75pt) and (2pt,0.75pt)}, tile size={(2pt, 2pt)}, tile transformation={rotate=-45}, code={\draw [line width=1.5pt] (0,0) -- (2pt,0);} }

\tikzset{
  tensor/.style={
    inner sep = 0.055cm,
    shape = circle,
    draw,
    fill
  },
  t/.style={
    inner sep = 0.03cm,
    shape = circle,
    draw,
    fill
  } 
    }

\title{Two-Dimensional Bialgebras and Quantum Groups: Algebraic Structures and Tensor Network Realizations}

\author[1]{Jos\'e Garre-Rubio}
\author[2]{Andr\'as Moln\'ar}
\author[1,3]{Germ\'an Sierra}

\affil[1]{\small Instituto de F\'isica Te\'orica, UAM/CSIC, C. Nicol\'as Cabrera 13-15, Cantoblanco, 28049 Madrid, Spain}
\affil[2]{\small University of Vienna, Faculty of Mathematics, Oskar-Morgenstern-Platz 1, 1090 Vienna, Austria}
\affil[3]{\small Kavli Institute for Theoretical Physics, University of California, Santa
Barbara, CA 93106, USA}
\date{}

\begin{document}

\maketitle

\begin{abstract}

We introduce a framework to define coalgebra and bialgebra structures on two-dimensional (2D) square lattices, extending the algebraic theory of Hopf algebras and quantum groups beyond the one-dimensional (1D) setting. Our construction is based on defining 2D coproducts through horizontal and vertical maps that satisfy compatibility and associativity conditions, enabling the consistent growth of vector spaces over lattice sites. We present several examples of 2D bialgebras, including group-like and Lie algebra-inspired constructions and a quasi-1D coproduct instance that is applicable to Taft-Hopf algebras and to quantum groups. The approach is further applied to the quantum group $U_q[su(2)]$, for which we construct 2D generalizations of its generators, analyze $q$-deformed singlet states, and derive a 2D R-matrix satisfying an intertwining relation in the semiclassical limit. Additionally, we show how tensor network states, particularly PEPS, naturally induce 2D coalgebra structures when supplemented with appropriate boundary conditions. Our results establish a local and algebraically consistent method to embed quantum group symmetries into higher-dimensional lattice systems, potentially connecting to the emerging theory of fusion 2-categories and categorical symmetries in quantum many-body physics.

\end{abstract}

\tableofcontents

\section{Introduction}
The theory of coalgebras, bialgebras, and quantum groups emerged as one of the most profound developments in modern algebra, fundamentally transforming our understanding of symmetry and algebraic structures. Quantum groups first appeared in the early 1980s in the work of physicists and mathematicians studying quantum integrable systems, with the quantized enveloping algebra of $\mathfrak{sl}_2$. The formal framework was established around 1985 when Vladimir Drinfeld and Michio Jimbo independently discovered quantized enveloping algebras of Kac-Moody algebras as non-commutative and non-cocommutative Hopf algebras. Drinfeld's seminal address at the International Congress of Mathematicians in 1986 brought quantum groups to worldwide attention and laid the theoretical foundations that continue to drive research today \cite{DrinfeldQG}.

At their core, quantum groups represent a remarkable generalization of classical Lie groups and Lie algebras through deformation theory. While compact groups and semisimple Lie algebras were long considered "rigid" objects that could not be deformed, Drinfeld's insight was that considering equivalent but larger structures - namely group algebras or universal enveloping algebras - allows for meaningful deformation while preserving essential algebraic properties \cite{etingof2024briefintroductionquantumgroups}. These deformations typically depend on a parameter $q$, reducing to classical universal enveloping algebras when $q = 1$. The mathematical foundation rests on Hopf algebra structures, which simultaneously encode both algebra and coalgebra properties, that encode how to multiply elements and define them in more lattice sites respectively, in a compatible manner, which generalize known structures in group theory.

The significance of quantum groups extends far beyond pure algebra, with profound connections to mathematical physics and topology. The quantum group $U_q[su(2)]$ plays a particularly important role. This quantum group emerges naturally as the symmetry algebra of integrable lattice models, including the XXZ spin chain and the Pasquier-Saleur Hamiltonian \cite{PASQUIER1990523}. The XXZ Hamiltonian exhibits quantum group symmetries that are crucial for its exact solvability \cite{roscilde_xxz}.

The representation theory of quantum groups reveals deep connections to conformal field theory and topological quantum field theory. The braiding matrices of Wess-Zumino-Witten models are precisely given by the representation theory of $U_q[su(2)]$. This connection established quantum groups as fundamental objects in the classification of rational conformal field theories and their associated modular tensor categories. On a higher level the category of representations of semisimple weak Hopf algebras are in one-to-one correspondence with multi-fusion categories \cite{bohm1999weak}. 

In recent years, these generalized algebraic structures has been considered as symmetries of one-dimensional quantum many-body systems \cite{Gaiotto_2015, Thorngren19,QuellaQGSPT,Garre22_MPOSYM}. In this context tensor networks (TN) plays an important role. TN representations encode quantum states or operators as contractions of networks of individual tensors, allowing for efficient description of entangled states and relevant operators while preserving important symmetries \cite{orus2013practical}. In Ref.~\cite{molnar22}, two of the authors and their collaborators established a duality between the representation of weak Hopf algebras and matrix product operators (MPO) based on their coalgebra structure. That, together with the MPO characterization of fusion categories of Ref.~\cite{Bultinck17A}, has completed the relation between MPO representations of algebraic and categorical structures in 1D.

The extension of these ideas to higher-dimensional systems presents both conceptual and technical challenges \cite{andruskiewitsch2006examples}. While one-dimensional quantum systems with quantum group symmetries are well understood through the quantum inverse scattering method and Bethe ansatz techniques, the generalization to two-dimensional lattice systems requires fundamentally new approaches. Recent work has begun to explore 2-categories as symmetries on quantum many-body physics \cite{inamura202521dlatticemodelstensor,bhardwaj2024gappedphases21dnoninvertible,bhardwaj2025gappedphases21dnoninvertible,Delcamp_2024,Inamura_2024,garreTN2DSYM}. However, the characterization of this categorical structures in terms of algebras and coalgebras are missing. This suggests that generalizations of quantum group structures to higher-dimensional systems may provide new insights into 2-categorical structures and their applications \cite{nastase_yangbaxter}.

The work presented in this paper represents the first systematic proposal for endowing two-dimensional lattice systems with coalgebra and bialgebra structures in a local and compatible manner. By extending the framework of Hopf algebras to two-dimensional lattices, we provide a new approach to constructing quantum operators and states that respects the geometric structure of the lattice while maintaining algebraic consistency. Our focus on the quantum group $U_q[su(2)]$ connects this work to the rich tradition of integrable quantum systems while opening new avenues for understanding higher-dimensional generalizations. The construction of tensor network representations for these structures provides a concrete computational framework that bridges abstract algebraic concepts with practical applications in quantum many-body theory. By providing explicit realizations of higher-dimensional coalgebraic structures, this work may contribute to the broader program of understanding categorification and its role in modern mathematics and physics.

The structure of the paper is as follows. We begin by reviewing the mathematical definition of Hopf algebras, whose structures we aim to generalize to higher-dimensional settings. In Section \ref{MainSec}, we present our proposal by defining two families of coalgebra maps that 'grow' column and row vectors. We also define the compatibility condition with an algebra structure and propose generalizations to other lattice geometries, including the 2D triangular and 3D square lattices. In Section \ref{SecExamp}, we construct several examples—some of them trivial, in the sense that they coincide with the conventional notions of 2D group-like and Lie-algebra-like objects, and others representing new constructions that can be applied to Taft–Hopf algebras and quantum groups. In Section \ref{PEPS}, we demonstrate that tensor network states in 2D naturally encode our two-dimensional coalgebra maps, and we construct the corresponding tensors for the examples. Finally, in Section \ref{secQG}, we study the case of the quantum group $U_q[su(2)]$ in detail and analyze its R-matrix. 

\section{Hopf Algebra Conditions: A Recap}

A Hopf algebra is an algebraic structure that combines the properties of algebras and coalgebras in a compatible way, supplemented by an additional operation called the antipode. Understanding the complete set of conditions required for a Hopf algebra provides essential background for the constructions developed in this work.

\textbf{Underlying Structure}: A Hopf algebra $H$ is defined over a field,  that we will consider to be just $\mathbb{C}$ for simplicity, and consists of a vector space over that field equipped with five fundamental operations that must satisfy specific compatibility conditions.

\textbf{Algebra Structure}: The algebra structure consists of:
\begin{itemize}
\item Multiplication map $\mu: H \otimes H \to H$, often written as $\mu(a \otimes b) = ab$
\item Unit map $\eta: \mathbb{C} \to H$, with $\eta(1) = 1_H$ denoting the multiplicative identity
\end{itemize}

These must satisfy associativity: $(ab)c = a(bc)$ for all $a,b,c \in H$, and unit properties: $1_H a = a 1_H = a$ for all $a \in H$.

\textbf{Coalgebra Structure}: The coalgebra structure consists of:
\begin{itemize}
\item Comultiplication map $\Delta: H \to H \otimes H$, written in Sweedler notation as $\Delta(h) = \sum_{(h)}h_{(1)} \otimes h_{(2)}$
\item Counit map $\varepsilon: H \to \mathbb{C}$
\end{itemize}

These must satisfy coassociativity: $(\Delta \otimes \text{id}) \circ \Delta = (\text{id} \otimes \Delta) \circ \Delta$, which in Sweedler notation becomes $\sum h_{(1)} \otimes h_{(2)(1)} \otimes h_{(2)(2)} = \sum h_{(1)(1)} \otimes h_{(1)(2)} \otimes h_{(2)}:= \sum h_{(1)} \otimes h_{(2)} \otimes h_{(3)}$. The counit satisfies $(\varepsilon \otimes \text{id}) \circ \Delta = (\text{id} \otimes \varepsilon) \circ \Delta = \text{id}$, or equivalently $\sum \varepsilon(h_{(1)})h_{(2)} = \sum h_{(1)}\varepsilon(h_{(2)}) = h$.

\textbf{Bialgebra Compatibility}: For the algebra and coalgebra structures to be compatible, forming a bialgebra, the following conditions must hold:
\begin{itemize}
\item The comultiplication $\Delta$ must be an algebra homomorphism: $\Delta(ab) = \Delta(a)\Delta(b)$ and $\Delta(1_H) = 1_H \otimes 1_H$
\item The counit $\varepsilon$ must be an algebra homomorphism: $\varepsilon(ab) = \varepsilon(a)\varepsilon(b)$ .
\end{itemize}

We note that the axiom on the counit implies that $\varepsilon(1_H) = 1$. Equivalently, one requires that the multiplication and unit are coalgebra homomorphisms.

\textbf{Antipode Condition}: The defining feature that elevates a bialgebra to a Hopf algebra is the existence of an antipode map $S: H \to H$ satisfying the fundamental antipode condition:
\begin{equation}
\sum S(h_{(1)})h_{(2)} = \sum h_{(1)}S(h_{(2)}) = \varepsilon(h)1_H
\end{equation}
for all $h \in H$.

\textbf{Antipode Properties}: When it exists, the antipode is unique and possesses several important properties:
\begin{itemize}
\item It is an anti-homomorphism of both algebras and coalgebras: $S(ab) = S(b)S(a)$ and $\Delta(S(h)) = \sum S(h_{(2)}) \otimes S(h_{(1)})$
\item For group-like elements $g$ (satisfying $\Delta(g) = g \otimes g$), we have $S(g) = g^{-1}$
\item For primitive elements $x$ (satisfying $\Delta(x) = x \otimes 1 + 1 \otimes x$), we have $S(x) = -x$
\end{itemize}

\textbf{Examples}: Classical examples include group algebras $\mathbb{C}[G]$ where $\Delta(g) = g \otimes g$ and $S(g) = g^{-1}$, and universal enveloping algebras $U(\mathfrak{g})$ where $\Delta(x) = x \otimes 1 + 1 \otimes x$ and $S(x) = -x$ for Lie algebra elements. 

These conditions collectively ensure that the category of $H$-modules forms a rigid monoidal category, providing the categorical foundation for the representation theory and the connections to tensor categories.

\section{Two-dimensional coproducts from row and column maps}\label{MainSec}

The coproduct described in the previous section can be seen as a map from a single point to a line, i.e. going from 0D to 1D. We want to generalize that construction to higher dimensions, in particular define coproducts from 1D to 2D.

We note that the naive approach of combining the same 1D coproduct in the $x$ and $y$ direction to construct a 2D operator is on the one hand, a trivial extension and on the other hand, only compatible for certain coproducts. We can check, see Sec.~\ref{Sec:extriv} for a proof, that this approach only works for (co)commutative coproducts which is only true for a few examples.

Let us introduce now our proposal. Let be $V$ a complex vector space. We first define a family of linear maps that take a column or row vector, of any size $n$, and grows it by doubling it to create a 2D object:
\begin{equation}
     \sq^n: V^{\otimes n}\to  V^{\otimes n}\otimes V^{\otimes n}   
\end{equation}

We will denote by $\sq^n_x$ the map corresponding to the horizontal 'coproduct' that grows column vectors horizontally. Similarly, we define the vertical coproduct as $\sq^n_y$ growing row vector on the vertical direction, that is:

$$ \sq^n_x \begin{pmatrix} *\\ \vdots\\ *\end{pmatrix} \in  \ Span \left \{  \begin{pmatrix} *\\ \vdots\\ *\end{pmatrix} \otimes \begin{pmatrix} *\\ \vdots\\ *\end{pmatrix} \right \} \ ,
\quad 
\sq^n_y(* \ \cdots \ *) \in Span \left\{  \begin{matrix} (* \ \cdots \ *)\\ \otimes \\ (* \ \cdots \ *)\end{matrix} \right \} \ .
$$
We interpret these maps as the way to grow the boundary. We require that these maps are associative in one direction:
\begin{equation}\label{1dquasiasso}
(\sq_i^n\otimes \id_n)\circ \sq_i^n = ( \id_n \otimes \sq_i^n)\circ \sq_i^n \ , \ i\in\{ x,y\}
\end{equation}
we call this equation quasi-1D associativity.

These maps allow us to define in a square lattice of size $n \times m$ a vector $v\in V$, that we denote by $\Sq^{n,m}(v)$. Importantly, the different ways of constructing these two-dimensional lattice vectors for a given system size must be compatible. In particular, the minimum square lattice size, $2\times 2$, gives rise to the $xy$-compatibility condition:
\begin{equation}\label{xycomp}
    \Sq^{2,2} \equiv \sq^2_x \circ \sq^1_y = \sq^2_y \circ \sq^1_x \ . 
\end{equation}
which pictorially can be drawn as

 $$
     
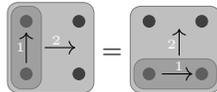
\begin{tikzpicture}[baseline=6pt]
    \foreach \x in {0,0.7}{ 
    \foreach \y in {0,0.7}{
    \node[tensor] at (\x,\y){};}}
    \draw[rounded corners, fill=gray, opacity=0.5]  (-0.2,-0.2) rectangle (0.2,0.9);
    \draw[rounded corners, fill=gray, opacity=0.5]  (-0.25,-0.25) rectangle (0.9,0.95);
    \draw[->] (0.25,0.35)--++(0.4,0);
    \draw[ ->] (0.0,0.15)--(0.0,0.55);
    \node[white] at (-0.075,0.35) {\tiny $1$};
    \node[white] at (0.4,0.45) {\tiny $2$};
    \end{tikzpicture}
    =
    \begin{tikzpicture}[baseline=6pt]
    \foreach \x in {0,0.7}{ 
    \foreach \y in {0,0.7}{
    \node[tensor] at (\x,\y){};}}
    \draw[rounded corners, fill=gray, opacity=0.5]  (-0.2,-0.2) rectangle (0.9,0.2);
    \draw[rounded corners, fill=gray, opacity=0.5]  (-0.25,-0.25) rectangle (0.95,0.9);
    \draw[->] (0.175,0)--++(0.4,0);
    \draw[ ->] (0.4,0.25)--++(0.0,0.35);
    \node[white] at (0.3,0.4) {\tiny $2$};
    \node[white] at (0.4,0.1) {\tiny $1$};
    \end{tikzpicture}
    $$
The associativity of $\sq$ allows to grow the vectors $\Sq$ in a compatible way:
\begin{align}
    \Sq^{n+1,m} = & \  \overbrace{\left (\id^{\otimes m} \otimes \cdots \otimes \id^{\otimes m} \otimes \sq^m_x \right )}^ {\# n}\circ \ \Sq^{n,m}   \\
   = & \ \left (\id^{\otimes m} \otimes \cdots \otimes \sq^m_x  \otimes \id^{\otimes m} \right )\circ \ \Sq^{n,m} \quad , \quad m\ge 2 \\
   = & \ \left ( \sq^m_x  \otimes \cdots \otimes \id^{\otimes m}   \otimes \id^{\otimes m} \right )\circ\  \Sq^{n,m} \ .
\end{align}
and also in the $y$-direction. 

The $xy$-compatibility of Eq.~\eqref{xycomp} can be generalized to any system size, that is, going from size $n\times m$ to $(n+1)\times( m+1)$ in two compatible ways:
\begin{equation}\label{Sqcomp}
    \Sq^{n+1,m+1} \equiv( \id \otimes \sq^{m+1}_x) \circ (\id\otimes \sq^n_y ) \circ \Sq^{n,m} = (\id \otimes \sq^{n+1}_y )\circ (\id \otimes \sq^m_x) \circ \Sq^{n,m} \ ,
\end{equation}
where we have used the co-associativity to avoid writing where the identity acts.

We note that any vector can be grown just by using $\sq^1$, for example in a $2\times 2$ as $(\sq^1_y\otimes\sq^1_y)\circ\sq^1_x$. However, we do not impose that to be equal to $\Sq$: we only allow to grow the full vertical and horizontal boundaries which only allows to obtain rectangular shapes.

Similarly to the case of one-dimensional Hopf algebras we can define counits. In this case we define a family of linear functionals:
\begin{equation}
    \epsilon_i^n: V^{\otimes n}\to \mathbb{C} \ , \ n\ge 1 \ , \ i\in \{ x,y\} \ ,
\end{equation}
that satisfies the following compatibility condition with $\sq^n$ (undoing it):
$$ (\epsilon_i^n\otimes \id )\circ \sq_i^n = (\id \otimes \epsilon_i^n )\circ \sq_i^n = \id , \ i\in \{ x,y\} \ ,  $$
which particularizes to the relation between the counit and $\Delta$ for $n=1$.
Let us summarize our proposal in a definition:
\begin{definition}[2D coalgebra] A two-dimensional coalgebra on a square lattice is given by a vector space $V$, where $V_{i,j}$ denotes the vector space associated with the site $(i,j)$, and a family of maps $(\sq^n_i,\epsilon^n_i), \ i\in \{x,y\}$ defined as follows
\begin{align*}
    \sq_x^n: \bigotimes_{j=1}^n V_{i,j}\to  \bigotimes_{j=1}^nV_{i,j}\otimes \bigotimes_{j=1}^nV_{i+1,j} \ , & \quad 
    \epsilon_x^n: \bigotimes_{j=1}^n V_{i,j} \to \mathbb{C} \ , \\
       \sq_y^m: \bigotimes_{i=1}^m V_{i,j}\to  \bigotimes_{i=1}^m V_{i,j}\otimes \bigotimes_{i=1}^m V_{i,j+1} \ , & \quad
    \epsilon_x^m: \bigotimes_{i=1}^m V_{i,j} \to \mathbb{C}  \ ,
\end{align*}
satisfying the following axioms:
\begin{itemize}
    \item Associativity axiom: 
    $$(\sq_i^n\otimes \id_n)\circ \sq_i^n = ( \id_n \otimes \sq_i^n)\circ \sq_i^n \ , \ i\in\{ x,y\} \ .$$
    \item Compatibility axiom: 
    $$\Sq^{n+1,m+1} \equiv( \id \otimes \sq^{m+1}_x) \circ (\id\otimes \sq^n_y ) \circ \Sq^{n,m} = (\id \otimes \sq^{n+1}_y )\circ (\id \otimes \sq^m_x) \circ \Sq^{n,m}\ .$$
    \item Counit axiom: 
    $$(\epsilon_i^n\otimes \id )\circ \sq_i^n = (\id \otimes \epsilon_i^n )\circ \sq_i^n = \id , \ i\in \{ x,y\}\ .$$
\end{itemize}
\end{definition}

\subsection{A compatible algebra structure: 2D bialgebras}

We now study the case when our space is both an algebra and a coalgebra: $V \equiv \A$, such that the
coproducts $\sq_i^n$ are homomorphism:
\begin{equation}
    \sq_{i}^n(a)\cdot \sq_{i}^n(b) = \sq_{i}^n(ab)\ , \ \forall a,b\in \A^{\otimes n} \ , i\in\{x,y\}
\end{equation}
where the multiplication is taken component-wise, i.e. $(a_1\otimes \cdots \otimes a_n) \cdot (b_1\otimes \cdots \otimes b_n) = a_1b_1\otimes \cdots \otimes a_nb_n$. Notice that this implies that the operators defined on the whole two-dimensional square lattice are also homomorphisms:

\begin{equation}
    \Sq^{n,m}(a)\cdot \Sq^{n,m}(b) = \Sq^{n,m}(ab)\ , \ \forall a,b\in \A \ , \ n,m\in \mathbb{N}
\end{equation}

Once we have this algebra structure we can also define a family of antipodes:
\begin{equation}
    S_i^n: \A^{\otimes n}\to \A^{\otimes n} \ , \ n\ge 1\ , i\in\{x,y\} \ ,
\end{equation}
satisfying 
$$  \mu_i (S_i^n\otimes \id)\circ \sq_i^n(a) = \mu_i (\id\otimes S_i^n)\circ \sq_i^n(a) = \epsilon_i^n(a) \ , i\in\{x,y\} \ ,$$
where $\mu_x$ or $\mu_y$ stands for the multiplication of two row or column algebra elements in $\A^{\otimes n}$ respectively. Altogether, the space $\A$, the maps $(\sq_i^n,\epsilon_i^n,S_i^n)$ with $i\in\{x,y\}$ and their compatibility conditions form a two-dimensional Hopf algebra.

\subsection{The dual two-dimensional algebra}

In the one-dimensional case, given a vector space $V$ with a coalgebra structure $(\Delta,\epsilon)$, an algebra can be defined on the space of linear functionals $V^*=\{f:V\to \mathbb{C}, \ f\ {\rm is \ linear} \}$. The product on this algebra is given by the coproduct on $V$: $(f\cdot g) (v) = (f\otimes g)\circ \Delta(v), \ f,g \in V^*$ and the unit is the counit since $(f\cdot \epsilon) (v) = (f\otimes \epsilon)\circ \Delta(v) = f(v), \ f \in V^*$. The co-associativity, $(\id\otimes \Delta)\circ \Delta = (\Delta \otimes \id)\circ \Delta $ guarantees that the product in $V^*$ is associative: $(f\cdot g) \cdot h = f\cdot (g \cdot h)$.

Similarly, we want to define an algebra as the dual of $(V,\sq)$. We first notice that we can extend the product functions evaluated on $V^{\otimes n}$ by
\begin{equation}\label{eq:produal}
    \mu_i(f_n \cdot g_n) (v_n) := (f_n \otimes g_n) \circ \sq_{i}^n(v_n) , \  f_n, g_n\in (V^{\otimes n })^* , v_n \in V^{\otimes n } \ , \in \{x,y\}
\end{equation}
where we can take either the $x$ or $y$ coproduct, which are in general different, by arranging the functions as rows or columns. The quasi-1D coassociativity Eq.~\eqref{1dquasiasso} guarantees the associativity of the different products of these functions.


This allows us to define the action of grid functions $f=\{f_{i,j}\}\in (V^{\otimes n\times m })^* $ on $v\in V $ via the map $\Sq$:
 
$$ f(v) := \bigotimes_{i,j} f_{i,j} \circ \Sq(v) \ .$$

This product of functions is associative in the sense that we can gather either row or column functions and obtain the same result. In particular for a $2\times 2$ square lattice
$$
\begin{pmatrix}f_{2,1}\\f_{1,1}\end{pmatrix}\cdot \begin{pmatrix}f_{2,2}\\ f_{1,2}  \end{pmatrix} (\sq^1_y(v))
=
\begin{matrix} (f_{2,1} , f_{2,2} )\\  \cdot \\(f_{1,1} , f_{1,2} ) \end{matrix} (\sq^1_x(v)) \ ,
$$
since the $xy$-compatibility is satisfied \eqref{xycomp}, $\sq^2_x \circ \sq^1_y = \sq^2_y \circ \sq^1_x $. 
For the sake of completeness let us write the conditions for $3\times 2$ lattice:
$$
\begin{pmatrix}f_{2,1}\\f_{1,1}\end{pmatrix}\cdot \begin{pmatrix}f_{2,2}\\ f_{1,2}  \end{pmatrix} \cdot \begin{pmatrix}f_{2,3}\\ f_{1,3}  \end{pmatrix} (\sq^1_y(v))
=
\begin{matrix} (f_{2,1} , f_{2,2} ,f_{2,3} )\\  \cdot \\(f_{1,1} , f_{1,2}, f_{1,3}) \end{matrix} (\sq^1_x)^2(v) \ .
$$

\subsection{Generalizations to the triangular and the 3D square lattices}

For a 3D square lattice there are three types of coproduct associated with each direction: $x,y,z$. The first compatibility equations comes from the growth only in two-dimensions, corresponding to the faces of a cube:
$$\sq^2_x \circ \sq^1_y = \sq^2_y \circ \sq^1_x \ , \ \sq^2_x \circ \sq^1_z = \sq^2_z \circ \sq^1_x\ , \ \sq^2_y \circ \sq^1_z = \sq^2_z \circ \sq^1_y \ .$$
The proper minimal 3D structure is a cube of size $2\times 2\times 2$ which imposes the $xyz$-compatibility:
$$ \sq^4_z \circ \sq^2_y \circ \sq^1_x =  \sq^4_y \circ \sq^2_z \circ \sq^1_x = \sq^4_x \circ \sq^2_y \circ \sq^1_z\ . $$
It is interesting to note that the non-trivial example provided in Section \ref{sec:pivot} can also be constructed in this example.

For a 2D triangular lattice with generating vectors $a, b$ and $c=a+b$ there are also three types of coproducts related to these directions. The first compatibility conditions are: 
$$\sq^2_a \circ \sq^1_b = \sq^2_b \circ \sq^1_a \ , \quad
\begin{tikzpicture}[baseline=6pt]
\node[tensor](a) at (0,0){};
\node[tensor](b) at (1,0){};
\node[tensor](c) at (60:1){};
\node[tensor](d) at (-60:1){};
\draw[dashed] (a)--(b)--(d);
\draw[->] (60:0.25)--++(60:0.5);
\node[] at (75:0.5) {\small $1$};
\draw[->] (-60:0.25)--++(-60:0.5);
\node[] at (-75:0.5) {\small $2$};
\draw[->] ($(-60:0.25)+(60:1)$)--++(-60:0.5);
\node[] at ($(-75:0.5)+(60:1)$) {\small $2$};
\end{tikzpicture}
=
\begin{tikzpicture}[baseline=6pt]
\node[tensor](a) at (0,0){};
\node[tensor](b) at (1,0){};
\node[tensor](c) at (60:1){};
\node[tensor](d) at (-60:1){};
\draw[dashed] (a)--(b)--(c);
\draw[->] (-60:0.25)--++(-60:0.5);
\node[] at (-75:0.5) {\small $1$};
\draw[->] (60:0.25)--++(60:0.5);
\node[] at (75:0.5) {\small $2$};
\draw[->] ($(60:0.25)+(-60:1)$)--++(60:0.5);
\node[] at ($(75:0.5)+(-60:1)$) {\small $2$};
\end{tikzpicture}
$$

 $$    \sq^2_a \circ \sq^1_c = \sq^2_c \circ \sq^1_a\ , \quad \begin{tikzpicture}[baseline=6pt]
\node[tensor] at (0,0){};
\node[tensor](a) at (1,0){};
\node[tensor](b) at (60:1){};
\draw[dashed] (a)--(b);
\draw[->] (0.2,0.2)--++(60:0.5);
\node[] at (70:0.5) {\small $1$};
\node[tensor](d) at ($(a)+(60:1)$){};
\draw[dashed] (d)--(a);
\draw[->] ($(0.2,0)+(60:1)$)--++(0.6,0);
\draw[->] (0.2,0)--++(0.6,0);
\node[](c) at (0.5,0.2) {\small $2$};
\node[] at ($(c)+(60:1)$) {\small $2$};
    \end{tikzpicture}
    =
\begin{tikzpicture}[baseline=6pt]
\node[tensor] at (0,0){};
\node[tensor](a) at (1,0){};
\node[tensor](b) at (60:1){};
\draw[dashed] (a)--(b);
\draw[->] (0.2,0.2)--++(60:0.5);
\draw[->] ($(0.2,0.2)+(1,0)$)--++(60:0.5);
\node[](c) at (70:0.5) {\small $2$};
\node[tensor](d) at ($(a)+(60:1)$){};
\draw[dashed] (d)--(b);
\draw[->] (0.2,0)--++(0.6,0);
\node[] at (0.5,0.2) {\small $1$};
\node[] at ($(c)+(1,0)$) {\small $2$};
\end{tikzpicture}
 $$

 $$ \sq^2_b \circ \sq^1_c = \sq^2_c \circ \sq^1_b \ , \quad    
 \begin{tikzpicture}[baseline=-6pt]
\node[tensor] at (0,0){};
\node[tensor](a) at (1,0){};
\node[tensor](b) at (-60:1){};
\draw[dashed] (a)--(b);
\draw[->] (0.2,-0.2)--++(-60:0.5);
\node[] at (-70:0.5) {\small $1$};
\node[tensor](d) at ($(a)+(-60:1)$){};
\draw[dashed] (d)--(a);
\draw[->] ($(0.2,0)+(-60:1)$)--++(0.6,0);
\draw[->] (0.2,0)--++(0.6,0);
\node[](c) at (0.5,0.2) {\small $2$};
\node[] at ($(c)+(-60:1)$) {\small $2$};
\end{tikzpicture}
    =
\begin{tikzpicture}[baseline=-6pt]
\node[tensor] at (0,0){};
\node[tensor](a) at (1,0){};
\node[tensor](b) at (-60:1){};
\draw[dashed] (a)--(b);
\draw[->] (0.2,-0.2)--++(-60:0.5);
\draw[->] ($(0.2,-0.2)+(1,0)$)--++(-60:0.5);
\node[](c) at (-70:0.5) {\small $2$};
\node[tensor](d) at ($(a)+(-60:1)$){};
\draw[dashed] (d)--(b);
\draw[->] (0.2,0)--++(0.6,0);
\node[] at (0.5,0.2) {\small $1$};
\node[] at ($(c)+(1,0)$) {\small $2$};
\end{tikzpicture}
    $$
The three growing directions have to be compatible to create the same vector for a full hexagon which results in the following compatible equations:

$$ \sq^3_c \circ \sq^2_a \circ \sq^1_b =  \sq^3_b \circ \sq^2_c \circ \sq^1_a = \sq^3_a \circ \sq^2_b \circ \sq^1_c \ ,$$
and diagrammatically can be expressed as:
$$
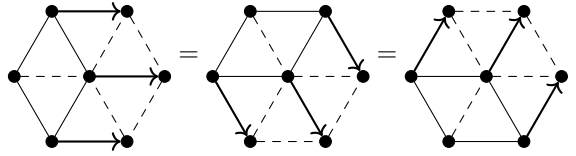
\begin{tikzpicture}[baseline=6pt]
\node[tensor](a) at (0,0){};
\node[tensor](b) at (1,0){};
\node[tensor](c) at (60:1){};
\node[tensor](d) at (-60:1){};
\draw[] (a)--(c)--(b)--(d)--(a);
\draw[dashed] (a)--(b);
\node[tensor](e) at ($(b)+(1,0)$){};
\draw[->,thick] (b)--(e);
\node[tensor](f) at ($(c)+(1,0)$){};
\draw[->,thick] (c)--(f);
\node[tensor](g) at ($(d)+(1,0)$){};
\draw[->,thick] (d)--(g);
\draw[dashed] (b)--(f)--(e)--(g)--(b);
\end{tikzpicture}
=
\begin{tikzpicture}[baseline=6pt]
\node[tensor](a) at (0,0){};
\node[tensor](b) at (1,0){};
\node[tensor](c) at (60:1){};
\node[tensor](d) at ($(b)+(60:1)$){};
\draw[] (a)--(b)--(d)--(c)--(a);
\draw[dashed] (c)--(b);
\node[tensor](e) at ($(a)+(-60:1)$){};
\draw[->,thick] (a)--(e);
\node[tensor](f) at ($(b)+(-60:1)$){};
\draw[->,thick] (b)--(f);
\node[tensor](g) at ($(d)+(-60:1)$){};
\draw[->,thick] (d)--(g);
\draw[dashed] (e)--(f)--(g);
\draw[dashed] (e)--(b)--(g);
\end{tikzpicture}
=
\begin{tikzpicture}[baseline=6pt]
\node[tensor](a) at (0,0){};
\node[tensor](b) at (1,0){};
\node[tensor](c) at (-60:1){};
\node[tensor](d) at ($(b)+(-60:1)$){};
\draw[] (a)--(b)--(d)--(c)--(a);
\draw[dashed] (c)--(b);
\node[tensor](e) at ($(a)+(60:1)$){};
\draw[->,thick] (a)--(e);
\node[tensor](f) at ($(b)+(60:1)$){};
\draw[->,thick] (b)--(f);
\node[tensor](g) at ($(d)+(60:1)$){};
\draw[->,thick] (d)--(g);
\draw[dashed] (e)--(f)--(g);
\draw[dashed] (e)--(b)--(g);
\end{tikzpicture}
\ .$$

We notice that a similar treatment can be used to construct coalgebras to grow the dice lattice (rhombille tiling). The triangular and dice lattice are the duals of the hexagonal and kagome lattice respectively. We leave for future work the development of coalgebras maps for these lattices which, regarding their dual covered by our approach, seems to be more appropriate to take plaquette-based maps instead of vertex one.

\section{Examples} \label{SecExamp}

In this section we propose several examples that satisfy our definition of 2D coproduct. We start by showing that using just the concatenation of 1D coproducts we obtain trivial 2D vectors in the sense that they are just come from the power of the 1d coproduct rearrange in a 2D lattice. Then, we construct a 2D coproduct based on 1D coproducts but not completely reducible to those and a proper 2D example without the connection to a 1D coproduct.  

\subsection{Trivial  2D coproduct as tensor product of two 1D coproducts}\label{Sec:extriv}

Let us start showing the following result:

\begin{proposition}

Let us define $\sq_i^n = \Delta_i\otimes \cdots \otimes \Delta_i $ in both the $x$ and $y$ axis with $\Delta_x:V_{i,j}\to V_{i,j}\otimes V_{i+1,j} $ and $\Delta_y:V_{i,j}\to V_{i,j}\otimes V_{i,j+1} $ that could be possibly different but defined in the same vector space $V$. The $xy$-compatibility condition, Eq.~\eqref{xycomp}, is satisfied if the following is true:
\begin{equation}\label{xy2d1d}
(\Delta_y\otimes \Delta_y)\circ \Delta_x = (\Delta_x\otimes \Delta_x)\circ \Delta_y \ .
\end{equation}
Then, the coproduct and counit are the same: $$\Delta_y = \Delta_x\equiv \Delta, \ \epsilon_x = \epsilon_y \ ,$$ 
and the coproduct is co-commutative:
$$\Delta(v) = \sum v_{(1)}\otimes v_{(2)} =  \sum v_{(2)}\otimes v_{(1)} \equiv \Delta^{op}(v) \ , \forall v\in V \ .$$
\end{proposition}

\begin{proof}
By applying $\epsilon_x\otimes \epsilon_y\otimes \id\otimes \id$ to Eq.~\eqref{xy2d1d} we get
$$ \Delta_y= \left(\id\otimes (\epsilon_y\otimes \id)\circ\Delta_x \right)\Delta_y \ .$$
If we further apply $\epsilon_y \otimes \id$ we obtain that $\id= (\epsilon_y\otimes \id)\circ \Delta_x$ which implies that
$$\epsilon_x=\epsilon_y\equiv \epsilon \ .$$
Let us now apply $\id\otimes \epsilon\otimes \epsilon\otimes \id$ to Eq.~\eqref{xy2d1d}, then we get 
$$\Delta_x = \Delta_y\equiv \Delta \ .$$
Similarly, applying $ \epsilon\otimes \id\otimes \id \otimes\epsilon$ we get that
$$ \Delta = \Delta^{op} \ .$$
\end{proof}

This implies that, by using the equivalence between MPS and coalgebras in 1D \cite{molnar22}, the resulting vector in the 2D lattice is a permutation invariant MPS, i.e., a trivial PEPS. The last statement can be seen by rewriting: 
$$\Sq^{2, 2}=(\Delta_x\otimes\id \otimes \id)\circ(\id\otimes \Delta_x)\circ \Delta_y= \Delta^3$$

There are two examples appearing naturally in physics where the coproduct is this trivial one:
\begin{itemize}
    \item Let us denote by $\Delta_g$ the 1D coproduct that acts as $\Delta_g(v) = v\otimes v$ for the basis elements $v\in V$ and extended by linearity. Taking $\sq^n_y = \sq^n_x = \Delta_g\otimes \cdots \otimes \Delta_g $ it is easy to see that it is a 'proper' 2D coproduct in the sense of the previous section and $\Sq^{n,m}(v)= v^{\otimes nm} $. This is nothing else but the group-like coproduct structure used in group representation theory, for example in on-site global symmetries like rotations $\bigotimes e^{i \theta\cdot \mathbf{J} }$.
    
    \item Let us denote by $\Delta_e$ the 1D coproduct that acts as $\Delta_e(a) = 1\otimes a+ a\otimes 1$ for $1,a\in \A$ where $1$ is the unit of the algebra. Taking $\sq^n_y = \sq^n_x = \Delta_e\otimes \cdots \otimes \Delta_e $ the resulting operators are $\Sq^{n,m}(a)= \sum_{i,j}a_{i,j}\otimes 1^{\otimes (nm-1)} $. This corresponds to the coproduct structure of Lie algebras like $su(2): S^z= \sum_i S^z_i$
\end{itemize}

\subsection{2D coproducts based quasi 1D coproducts}

Let us consider two 1D coproducts $\Delta_x$ and $\Delta_x$ defined in the same vector space $V$. We define the 2D coproducts as follows:
\begin{itemize}
    \item $ \sq^n_y=\Delta_y\otimes \cdots \otimes \Delta_y $ which is defined in the whole space $V^{\otimes n}$.
    \item $\sq^n_x$ is only defined in ${\rm Im}(\Delta^{n-1}_y)$ and its action reads
    $$ \sq^n_x(\Delta^{n-1}_y(v))= (\Delta^{n-1}_y\otimes \Delta^{n-1}_y) \Delta_x(v) .$$
\end{itemize}
With these definitions is clear that Eq.~\eqref{Sqcomp} is satisfied and that the 2D operators are:
$$ \Sq^{n,m}(v) = (\underbrace{\Delta^{n-1}_y\otimes \cdots \otimes \Delta^{n-1}_y}_{\# m}) \circ \Delta^{m-1}_x \ . $$
The counit in the $x$-direction is simply $\epsilon^n_y = \epsilon^{\otimes n}_y$ and $\epsilon^n_x$ is also only defined in ${\rm Im}(\Delta^{n-1}_y)$ as $\epsilon^n_x(\Delta^{n-1}_y(v))=\epsilon_x(v)$. The role of $x$ and $y$ can be interchange in the construction.

In what follows, we show two examples of this construction.
\begin{itemize}
    \item  {\bf 2D Group-like coproduct in one direction.} We will take as an input two 1D coproducts: $\Delta_g$ defined above and an arbitrary coproduct in 1D such that $\Delta(v) = \sum_{i} v^i_{(1)}\otimes v^i_{(2)}$, where we write explicitly the index of the sum $i$. Then we define $\sq^n_y = \Delta_g\otimes \cdots \otimes \Delta_g $ and 
$$\sq^n_x \begin{pmatrix} v\\ \vdots \\v\end{pmatrix}  = \sum_{i} \begin{pmatrix} v^i_{(1)}\\\vdots \\ v^i_{(1)} \end{pmatrix} \otimes \begin{pmatrix} v^i_{(2)}\\\vdots \\ v^i_{(2)}\end{pmatrix} \ ,$$
where $\sq^1_x = \Delta$ but $\sq^n_x \neq \Delta\otimes \cdots \otimes \Delta$. Then, the $xy$-compatibility condition is satisfied and reads;
$$ \sq^2_x\circ \sq^1_y (v) = \sq^2_x \begin{pmatrix} v\\ v\end{pmatrix} = \sum_i \begin{pmatrix} v^i_{(1)} \\ v^i_{(1)} \end{pmatrix} \otimes \begin{pmatrix} v^i_{(2)} \\ v^i_{(2)}\end{pmatrix} = \sq^2_y (\Delta(v)) = \sq^2_y \circ \sq^1_x(v) \ .$$
So that in arbitrary sizes we get:
$$ \Sq^{n,m}(v) = \sum_i \begin{pmatrix} v^i_{(1)}\\\vdots \\ v^i_{(1)} \end{pmatrix} \otimes \begin{pmatrix} v^i_{(2)}\\\vdots \\ v^i_{(2)}\end{pmatrix}\otimes \cdots \otimes \begin{pmatrix} v^i_{(n)}\\\vdots \\ v^i_{(n)}\end{pmatrix}$$

Similarly, the role of the $x$ and $y$ directions can be interchanged.

\item {\bf 2D Lie algebra-like coproduct in one direction}

We denote by $\Delta_e$ the 1D coproduct of a bialgebra $\mathcal{A}$ that acts as $\Delta_e(v) = 1\otimes v+ v\otimes 1$ for all $v\in V\equiv \mathcal{A}$. Similarly, we define $\sq^n_e (v_n) = 1_n\otimes v_n+ v_n\otimes 1_n$, $1_n,v_n \in \A^{\otimes n}$ where we can see that $\sq^n_e \neq \Delta_e\otimes \cdots \otimes \Delta_e$.
Let's take $\sq^n_e = \sq^n_y$ and $\sq^n_x = \Delta\otimes \cdots \otimes \Delta $ where $\Delta$ is an arbitrary coproduct in 1D with $\Delta(1)= 1\otimes 1$. 
Then, the $xy$-compatibility condition reads;
$$ \sq^2_x\circ \sq^1_y (v) = \sq^2_x\left( \begin{pmatrix} 1\\ v\end{pmatrix} +\begin{pmatrix}   v\\ 1\end{pmatrix}\right ) = \begin{pmatrix} 1_2\\ \Delta(v)\end{pmatrix} +\begin{pmatrix} \Delta(v) \\ 1_2 \end{pmatrix} = \sq^2_y (\Delta(v)) = \sq^2_y\circ \sq^1_x (v)$$

Then
$$ \Sq^{n,m}(v) = \begin{pmatrix} \Delta^n(v) \\ 1 \\ \vdots \\ 1 \end{pmatrix} +
\begin{pmatrix} 1   \\ \Delta^n(v) \\ \vdots \\ 1 \end{pmatrix} + \cdots 
\begin{pmatrix} 1   \\ \vdots \\ 1 \\ \Delta^n(v) \end{pmatrix}
$$
Similarly, the role of the $x$ and $y$ directions can be interchanged.

\end{itemize}

\subsection{Cross-like 2D algebra}
Let us consider an algebra $\A$ that is spanned by 6 vectors $1,t,b,r,l,v$. The goal is to define a 2D coproduct in such a way we arrive at the following 2D operator from $v$:
$$ \Sq^{n,m}(v)  = \sum_{\substack{1\le i \le n \\ 1\le j \le m}} 
\begin{tikzpicture}[baseline=(current bounding box.center)]
\node (mat) at (0,0) {
  $    \begin{pmatrix}
           & & & t & & &  \\
           &1 & & \vdots & & 1 & \\
           & & & t & & & \\
           l& \cdots & l& v &r &\cdots &r \\
           & & & b & & & \\
           &1 & & \vdots & & 1& \\
           & & & b & & &
    \end{pmatrix}$
};
\draw[<-, thick] (2.3,0) --++ (0.4,0) node[right] {$j$};
\draw[<-, thick] (0,1.8) --++ (0,0.3) node[above] {$i$};
\end{tikzpicture}
$$

The following compatible actions of the coproducts, defined only on the vectors below, do the job:

\begin{align*}  
    \sq^n_x \begin{pmatrix}
t \\ \vdots \\ t \\ v \\ b \\ \vdots \\ b
\end{pmatrix}
\hspace{-15pt}
\begin{array}{l}
\left.
\begin{array}{l}
\ \\ \ \\ \ 
\end{array}
\right\}
\# p
\\
\ \\ \ \\
\left.
\begin{array}{c}
\ \\ \ \\ \ 
\end{array}
\right\}
\# n-p-1
\end{array}
         = &
         \begin{pmatrix}
           t \\
           \vdots \\
           t \\
           v \\
           b \\
           \vdots \\
           b
         \end{pmatrix}
         \otimes
         \begin{pmatrix}
           1 \\
           \vdots \\
           1 \\
           r \\
           1 \\
           \vdots \\
           1
         \end{pmatrix}
         +
        \begin{pmatrix}
           1 \\
           \vdots \\
           1 \\
           l \\
           1 \\
           \vdots \\
           1
         \end{pmatrix}
         \otimes
         \begin{pmatrix}
           t \\
           \vdots \\
           t \\
           v \\
           b \\
           \vdots \\
           b
         \end{pmatrix} \ , \
         n\ge 1 \ , \ 0\le p<n \ , \\
\sq^n_x     \begin{pmatrix}
           1 \\
           \vdots \\
           1 \\
           r/l \\
           1 \\
           \vdots \\
           1
         \end{pmatrix}
         = &
        \begin{pmatrix}
           1 \\
           \vdots \\
           1 \\
           r/l \\
           1 \\
           \vdots \\
           1
         \end{pmatrix}
         \otimes
         \begin{pmatrix}
           1 \\
           \vdots \\
           1 \\
           r/l \\
           1 \\
           \vdots \\
           1
         \end{pmatrix} \ , 
         \ n\ge 2 \ , \\
    \sq^n_y ( l \ \cdots \ l \ v \ r \ \cdots \ r ) =& \begin{matrix} ( 1 \ \cdots \ 1 \ t \ 1 \ \cdots \ 1 ) \\ \otimes\\ ( l \ \cdots \ l \ v \ r \ \cdots \ r )
    \end{matrix} +
    \begin{matrix} ( l \ \cdots \ l \ v \ r \ \cdots \ r ) \\ \otimes \\( 1 \ \cdots \ 1 \ b \ 1 \ \cdots \ 1 )
    \end{matrix} \ ,\\
        \sq^n_y ( 1 \ \cdots \ 1 \ t/b \ 1 \ \cdots \ 1 ) =& \begin{matrix} ( 1 \ \cdots \ 1 \ t/b \ 1 \ \cdots \ 1 ) \\ \otimes\\ ( 1 \ \cdots \ 1 \ t/b \ 1 \ \cdots \ 1 )
    \end{matrix} \ .
\end{align*}
The counits are defined on the image of the coproducts as follows:
\begin{align*}
\epsilon^n_x(t^{\otimes p}\otimes v\otimes b^{\otimes {n-p-1}}) & = \epsilon^n_y(l^{\otimes p}\otimes v\otimes r^{\otimes {n-p-1}}) = 0 \ , n\ge 1 \ , \ 0\le p<n \ ,\\
\epsilon^n_i(1^{\otimes p}\otimes g\otimes 1^{\otimes {n-p-1}}) & = 1\ ,  \ g\in \{t,b,l,r\} \ , i\in\{x,y\} \ .\\
\end{align*}
An option to define the action of this coproduct in the whole vector space $\A$ is to consider as basis $1,t,b,r,l,v$ and then define the 2d coproduct on $1,t,b,r,l$ to behave as group-like elements. We notice that the application of this example to a quantum group structure does not work since the coproduct does not behave as homomorphism of the algebra. The following example does provide a proper 2D bialgebra for quantum groups.

\subsection{The main example: a non-local 1D coproduct}\label{sec:pivot}

Let us consider a 1D bialgebra with elements $a,b,v$ such that
\begin{equation} \label{coalprim} 
 \Delta(v)= a\otimes v + v\otimes b , \ \Delta(a)= a\otimes a , \ \Delta(b)= b\otimes b 
 \end{equation}
so that
$$ \Delta^n(v)= \sum a^{\otimes k} \otimes v \otimes b^{\otimes n-k} \ .$$
Then we can generalize this structure in a two-dimensional square lattice as follows:

\begin{equation}\label{defpiv}
     \Sq^{n,m}(v)  = \sum_{\substack{1\le i \le n \\ 1\le j \le m}} 
\begin{tikzpicture}[baseline=(current bounding box.center)]
     \node (mat) at (0,0) { $
    \begin{pmatrix}
           & & & b & & &  \\
           &b & & \vdots & & b & \\
           & & & b & & & \\
           a & \cdots & a& v &b &\cdots &b \\
           & & & a & & & \\
           &a & & \vdots & & a& \\
           & & & a & & &
    \end{pmatrix}$
    };
\draw[<-, thick] (2.3,0) --++ (0.4,0) node[right] {$j$};
\draw[<-, thick] (0,1.8) --++ (0,0.3) node[above] {$i$};
\end{tikzpicture}
\ ,
    \Sq^{n,m}(a/b) = (a/b)^{\otimes nm}
\end{equation}
To do so we define the coproducts in the following way: 

\begin{align*}  
    \sq^n_x \begin{pmatrix}
b \\ \vdots \\ b \\ v \\ a \\ \vdots \\ a
\end{pmatrix}
\hspace{-15pt}
\begin{array}{l}
\left.
\begin{array}{l}
\ \\ \ \\ \ 
\end{array}
\right\}
\# p
\\
\ \\ \ \\
\left.
\begin{array}{c}
\ \\ \ \\ \ 
\end{array}
\right\}
\# n-p-1
\end{array}
         = &
         \begin{pmatrix}
           b \\
           \vdots \\
           b \\
           v \\
           a \\
           \vdots \\
           a
         \end{pmatrix}
         \otimes
         \begin{pmatrix}
           b \\
           \vdots \\
           b \\
           b \\
           a \\
           \vdots \\
           a
         \end{pmatrix}
         +
        \begin{pmatrix}
           b \\
           \vdots \\
           b \\
           a \\
           a \\
           \vdots \\
           a
         \end{pmatrix}
         \otimes
         \begin{pmatrix}
           b \\
           \vdots \\
           b \\
           v \\
           a \\
           \vdots \\
           a
         \end{pmatrix} \ , \
         n\ge 1 \ , \ 0\le p<n \\
\sq^n_x     \begin{pmatrix}
           b \\
           \vdots \\
           b \\
           b/a \\
           a \\
           \vdots \\
           a
         \end{pmatrix}
         = &
        \begin{pmatrix}
           b \\
           \vdots \\
           b \\
           b/a \\
           a \\
           \vdots \\
           a
         \end{pmatrix}
         \otimes
         \begin{pmatrix}
           b \\
           \vdots \\
           b \\
           b/a \\
           a \\
           \vdots \\
           a
         \end{pmatrix} \ , 
         \ n\ge 1 \\
    \sq^n_y ( a \ \cdots \ a \ v \ b \ \cdots \ b ) =& \begin{matrix} ( b \ \cdots \ b \ b \ b \ \cdots \ b ) \\ \otimes\\ ( a \ \cdots \ a \ v \ b \ \cdots \ b )
    \end{matrix} +
    \begin{matrix} ( a \ \cdots \ a \ v \ b \ \cdots \ b ) \\ \otimes \\( a \ \cdots \ a \ a \ a \ \cdots \ a )
    \end{matrix} \\
        \sq^n_y ( a/b \ \cdots  \ a/b  ) =& \begin{matrix} ( a/b \ \cdots  \ a/b  ) \\ \otimes\\ ( a/b \ \cdots  \ a/b  )
    \end{matrix} 
\end{align*}
We note that $\sq^n_{x}=\Delta\otimes\cdots \otimes \Delta $ for any $n$, but in the $y$ direction only $\sq^1_{y}=\Delta$. Moreover we notice that if we order the elements of the 2D lattice from left to right and jumping on the last row element to the upper next row we can see that the operator is a rearrangement of the 1D coproduct in this 2D lattice. For example in a $3\times 3$ square:
\begin{equation}\label{horiorder}  \begin{pmatrix}7& 8 &9\\4 & 5&6\\1& 2&3  \end{pmatrix} \  \rightarrow \ 
\begin{tikzpicture}[scale=1]
    \draw[very thick] (-1.411, 2.258) -- (0.282, 2.258) -- (-1.411, 2.822) -- (0.282, 2.822) -- (-1.411, 3.387) -- (0.282, 3.387);
    \node[disk normal] at (-1.411, 2.258) {};
    \node[disk normal] at (-0.564, 2.258) {};
    \node[disk normal] at (0.282, 2.258) {};
    \node[disk normal] at (-1.411, 2.822) {};
    \node[disk normal] at (-0.564, 2.822) {};
    \node[disk normal] at (0.282, 2.822) {};
    \node[disk normal] at (0.282, 3.387) {};
    \node[disk normal] at (-0.564, 3.387) {};
    \node[disk normal] at (-1.411, 3.387) {};
\end{tikzpicture} \ .
\end{equation}
Then the coalgebra has a compatible 2D algebra structure inherited by the 1d bialgebra. The identification is as follows:
$$  \Sq^{L_x , L_y}_{i,j} = \Delta_{L_x(j-1)+i}\ , \ i=1,\dots, L_x \ , j=1,\dots, L_y \ .$$
We emphasize that this arrangement is not local in the sense that there is a jump of length the horizontal system size to go to the next row.

There are two important examples of bialgebras that fits in this setting. One is the family of Taft Hopf algebra and the other are quantum groups. We deal here with the former and in section \ref{secQG} with the later in detailed.

\subsubsection{2D Taft-Hopf algebra}
Let $\omega\in \mathbb{C}$ be an $n$-th root of unity: $\omega^n =1$ where $n\in\mathbb{N}$. The \emph{Taft-Hopf algebra} \cite{TAFT} is the finite-dimensional non-semisimple Hopf algebra generated as an algebra by two elements, $g$ and $x$, subject to the relations
    \[ g^n  = 1, \quad x^n = 0, \quad xg = \omega gx, \]
    for which the coalgebra structure is defined by the extensions of the linear maps:
    \[
    \Delta(g) = g\otimes g,\quad \Delta(x) = 1 \otimes x + x \otimes g, \quad \varepsilon(g) = 1,\quad \varepsilon(x) = 0,
    \]
Since the coproduct has the structure needed for the previous 2D generalization we can propose a 2D bialgebra based on the Taft Hopf algebra with the following 2D opertators:

\begin{align}\label{eqSqx}
\Sq^{n,m}(x)  &= \sum_{\substack{1\le i \le n \\ 1\le j \le m}} \hat{x}_{i,j}, \quad \hat{x}_{i,j} =
\begin{tikzpicture}[baseline=(current bounding box.center)]
\node (mat) at (0,0) { $
     \begin{pmatrix}
           & & & g & & &  \\
           &g & & \vdots & & g & \\
           & & & g & & & \\
           1& \cdots & 1& x &g &\cdots &g  \\
           & & & 1 & & & \\
           &1 & & \vdots & & 1& \\
           & & & 1 & & &
        \end{pmatrix}$ };
\draw[<-, thick] (2.3,0) --++ (0.4,0) node[right] {$j$};
\draw[<-, thick] (0,1.8) --++ (0,0.3) node[above] {$i$};
\end{tikzpicture}
 \\
         \label{eqSqg}
    \Sq^{n,m}(g) &=
             \begin{pmatrix}
           g&\cdots &g  \\
           \vdots& g& \vdots \\
          g & \cdots &g 
         \end{pmatrix} \equiv g^{\otimes nm}  \ .
\end{align}

\subsubsection{A generalization from another point of view} 
The previous arrangement of any 1D coalgebra with the coproduct given in Eq.~\eqref{coalprim} into a 2D lattice as \eqref{horiorder} can be modified. The generalization is as follows. Let us denote the term in the sum of $\Sq^{n,m}(v)$ that contains the $v$ in position $(i,j)$ as $\hat{v}_{i,j}$. Let us take an angle $\theta$ and assign $b$ to the lattice positions that are covered by the disk section between the angles $[\theta, \theta+\pi)$ with center in ${i,j}$, and the rest with $a$. The initial proposal of Eq.~\eqref{defpiv} corresponds to $\theta=0$. We note that since the lattice is discrete there are several angles that give rise to the same operators for the a given system size $n\times m$. As an example, let us take $\theta= \pi/4$ so in a $4\times 4$ lattice we have this term:
$$\hat{v}_{2,3} = 
    \begin{pmatrix}
           b & b & b& a  \\
           b&v &a & a \\
           a& a& a& a  \\
           a& a & a& a \\
    \end{pmatrix} \ ,$$
that does no appear in Eq.~$\eqref{defpiv}$. Importantly, this generalization can also be seen as a non-local rearrangement of the 1D coproduct in the 2D lattice by ordering the sites following the angle $\theta$. Then, the compatible algebra structure is also preserved in this 2D generalization. In the case of $\theta=\pi/4$ we can fill the 2D lattice as follows:
$$
\begin{tikzpicture}[scale=1]
    \node[disk large] at (-3.669, 0) {};
    \node[disk large] at (-2.822, 0) {};
    \node[disk large] at (-1.976, 0) {};
    \node[disk large] at (-1.129, 0) {};
    \node[disk large] at (-3.669, 0.847) {};
    \node[disk large] at (-2.822, 0.847) {};
    \node[disk large] at (-1.976, 0.847) {};
    \node[disk large] at (-1.129, 0.847) {};
    \node[disk large] at (-1.129, 1.693) {};
    \node[disk large] at (-3.669, 1.693) {};
    \node[disk large] at (-2.822, 1.693) {};
    \node[disk large] at (-1.976, 1.693) {};
    \node[disk large] at (-1.129, 2.54) {};
    \node[disk large] at (-1.976, 2.54) {};
    \node[disk large] at (-2.822, 2.54) {};
    \node[disk large] at (-3.669, 2.54) {};
    \draw[-Latex] (-1.976, 0) -- (-1.411, 0.564);
    \draw[-Latex] (-2.822, 0) -- (-2.258, 0.564);
    \draw[-Latex] (-1.976, 0.847) -- (-1.411, 1.411);
    \draw[-Latex] (-1.976, 1.693) -- (-1.411, 2.258);
    \draw[-Latex] (-2.822, 0.847) -- (-2.258, 1.411);
    \draw[-Latex] (-3.669, 0) -- (-3.104, 0.564);
    \draw[-Latex] (-3.669, 0.847) -- (-3.104, 1.411);
    \draw[-Latex] (-2.822, 1.693) -- (-2.258, 2.258);
    \draw[-Latex] (-3.669, 1.693) -- (-3.104, 2.258);
    \draw[-Latex] (-1.411, 0.564) -- (-0.847, 1.129);
    \draw[-Latex] (-1.411, 1.411) -- (-0.847, 1.976);
    \draw[-Latex] (-1.411, 2.258) -- (-0.847, 2.822);
    \draw[-Latex] (-2.258, 2.258) -- (-1.693, 2.822);
    \draw[-Latex] (-3.104, 2.258) -- (-2.54, 2.822);
    \draw[-Latex] (-1.411, -0.282) -- (-0.847, 0.282);
    \draw (-2.258, 1.411) -- (-1.976, 1.693);
    \draw (-3.104, 0.564) -- (-2.822, 0.847);
    \draw (-3.104, 1.411) -- (-2.822, 1.693);
    \draw (-2.258, 0.564) -- (-1.976, 0.847);
    \draw (-2.822, 0) -- (-3.104, -0.282);
    \draw (-3.669, 0) -- (-3.951, -0.282);
    \draw (-3.669, 0.847) -- (-3.951, 0.564);
    \draw (-3.669, 1.693) -- (-3.951, 1.411);
    \draw[shift={(-3.757, 2.467)}, xscale=1.065, yscale=1.086, Latex-] (0, 0) -- (-0.282, -0.282);
    \draw[black!75, densely dotted] (-1.976, 0) -- (-2.258, -0.282);
    \draw[black!75, densely dotted] (-0.847, 0.282) .. controls (-0.564, 0.564) and (-0.706, 0.564) .. (-0.964, 0.423) .. controls (-1.223, 0.282) and (-1.599, 0) .. (-1.858, -0.212) .. controls (-2.117, -0.423) and (-2.258, -0.564) .. (-2.258, -0.282);
    \draw[black!75, densely dotted] (-0.847, 1.129) .. controls (-0.564, 1.411) and (-0.847, 1.411) .. (-1.27, 1.129) .. controls (-1.693, 0.847) and (-2.258, 0.282) .. (-2.611, -0.071) .. controls (-2.963, -0.423) and (-3.104, -0.564) .. (-3.104, -0.282);
    \draw[black!75, densely dotted] (-0.847, 1.976) .. controls (-0.564, 2.258) and (-0.988, 2.117) .. (-1.576, 1.67) .. controls (-2.164, 1.223) and (-2.916, 0.47) .. (-3.363, 0.024) .. controls (-3.81, -0.423) and (-3.951, -0.564) .. (-3.951, -0.282);
    \draw[black!75, densely dotted] (-0.847, 2.822) .. controls (-0.564, 3.104) and (-0.847, 3.104) .. (-1.458, 2.634) .. controls (-2.07, 2.164) and (-3.01, 1.223) .. (-3.551, 0.753) .. controls (-4.092, 0.282) and (-4.233, 0.282) .. (-3.951, 0.564);
    \draw[black!75, densely dotted] (-1.693, 2.822) .. controls (-1.411, 3.104) and (-1.834, 2.963) .. (-2.328, 2.611) .. controls (-2.822, 2.258) and (-3.387, 1.693) .. (-3.739, 1.411) .. controls (-4.092, 1.129) and (-4.233, 1.129) .. (-3.951, 1.411);
    \draw[black!75, densely dotted] (-2.54, 2.822) .. controls (-2.258, 3.104) and (-2.681, 2.963) .. (-3.034, 2.752) .. controls (-3.387, 2.54) and (-3.669, 2.258) .. (-3.881, 2.117) .. controls (-4.092, 1.976) and (-4.233, 1.976) .. (-3.951, 2.258);
\end{tikzpicture} \ .
$$
Similarly, coproducts maps can be defined. 

The way in which the assignment of the elements is done, by a disk covering certain angle, in the 2D lattice allows for its application in the continuum. We left for future work this problem.

\section{Tensor network states create coalgebras}\label{PEPS}

In Ref.~\cite{molnar22}, the tensor network representation of coalgebras and bialgebras is developed. Concretely, for a given coalgebra and a representation of it, a matrix product state with boundaries is associated and vice-versa. 

It is clear that projected entangled pair states (PEPS) \cite{Verstraete04} are vectors that can be defined for any system size and as such, they should carry a 2D coalgebra structure. In this section, we show that given a PEPS and a set of allowed boundaries, a 2D coproduct can be defined from it.

The ingredients to construct a 2D coalgebra from a PEPS with a boundary is as follows. Let us consider a PEPS tensor $B\in V\otimes (\mathbb{C}^D)^{\otimes 4}, \ B=\sum_{i,j} B^j_{i_1,i_2,i_3,i_4} \ket{i_1,i_2,i_3,i_4}$ where $j$ runs into the basis of the physical space $V$ and $i$ denotes the index of the virtual levels with bond dimension $D$. We also need a set of MPS tensors $A_s\in \mathbb{C}^D\otimes (\mathbb{C}^\chi)^{\otimes 2} $ where $s\in \{l,r,t,b\}$ and stands for left, right, top, bottom and denotes in which part of the boundary of the PEPS the tensor $A$ is placed and the virtual dimension of $A$ is denoted by $\chi$.

For every $v\in V$ we will associate a matrix $b_v \in \mathcal{B} \subset \mathcal{M}_\chi$ (the set of boundaries) and we define:

\begin{equation}
    \ket{v} = \sum_{i,j} \Tr[A^{i_1}_l A^{i_2}_t A^{i_3}_r A^{i_4}_b b_v] B^j_{i_1,i_2,i_3,i_4} \ket{j} =
\begin{tikzpicture}[scale=1]
    \draw[RoyalBlue!80, very thick] (-1.411, 1.129) -- (-1.411, 1.693);
    \draw[very thick] (-2.54, 1.129) -- (-0.282, 1.129) -- (-0.564, 1.129);
    \draw[very thick] (-1.976, 0.282) -- (-0.847, 1.976);
    \node[disk normal] at (-1.411, 1.129) {};
    \draw[RedOrange, very thick] (-0.847, 1.976) .. controls (-1.599, 2.164) and (-2.164, 1.881) .. (-2.54, 1.129);
    \draw[RedOrange, very thick] (-2.54, 1.129) .. controls (-2.54, 0.753) and (-2.352, 0.47) .. (-1.976, 0.282);
    \draw[RedOrange, very thick] (-1.976, 0.282) .. controls (-1.223, 0.282) and (-0.659, 0.564) .. (-0.282, 1.129);
    \draw[RedOrange, very thick] (-0.847, 1.976) .. controls (-0.47, 1.787) and (-0.282, 1.505) .. (-0.282, 1.129);
    \node[disk normal] at (-2.54, 1.129) {};
    \node[disk normal] at (-1.976, 0.282) {};
    \node[disk normal] at (-0.282, 1.129) {};
    \node[disk normal] at (-0.847, 1.976) {};
    \node[anchor=center] at (-1.369, 0.902) {$B$};
    \node[anchor=center] at (0.006, 1.079) {$A_r$};
    \node[anchor=center] at (-0.665, 2.092) {$A_t$};
    \node[anchor=center] at (-2.113, 0.085) {$A_b$};
    \node[anchor=center] at (-2.848, 1.149) {$A_l$};
    \node[disk normal] at (-2.344, 0.583) {};
    \node[anchor=center] at (-2.559, 0.498) {$b_v$};
\end{tikzpicture} \ ,
\end{equation}
where we impose that this map $b_v\to v$ is injective.
This construction is naturally generalized to any system size, for example in a $2\times 2$ lattice:
\begin{equation}
\Sq^{2,2}(v)=  
\begin{tikzpicture}[scale=1]
    \draw[very thick] (-0.327, 3.372) -- (-2.002, 1.395);
    \draw[very thick] (-0.586, 1.427) -- (1.127, 3.406);
    \draw[very thick] (-1.411, 2.822) -- (1.129, 2.822);
    \draw[very thick] (-1.976, 1.976) -- (0.282, 1.976);
    \draw[RoyalBlue!80, very thick] (-0.792, 2.822) -- (-0.793, 3.164);
    \draw[RoyalBlue!80, very thick] (-1.51, 1.976) -- (-1.513, 2.338);
    \draw[RoyalBlue!80, very thick] (-0.111, 1.976) -- (-0.107, 2.329);
    \draw[RoyalBlue!80, very thick] (0.622, 2.822) -- (0.613, 3.15);
    \node[disk normal] at (-0.792, 2.822) {};
    \node[disk normal] at (-1.51, 1.976) {};
    \node[disk normal] at (-0.111, 1.976) {};
    \node[disk normal] at (0.622, 2.822) {};
    \draw[RedOrange, very thick] (-0.327, 3.372) .. controls (-0.745, 3.38) and (-1.101, 3.196) .. (-1.395, 2.822);
    \draw[RedOrange, very thick] (-1.4, 2.822) -- (-1.976, 1.976);
    \draw[RedOrange, very thick] (-1.972, 1.981) .. controls (-2.066, 1.854) and (-2.068, 1.669) .. (-1.976, 1.425);
    \draw[RedOrange, very thick] (-1.977, 1.425) -- (-0.579, 1.435);
    \draw[RedOrange, very thick] (-0.579, 1.435) .. controls (-0.215, 1.437) and (0.072, 1.617) .. (0.282, 1.976);
    \draw[RedOrange, very thick] (0.28, 1.976) .. controls (0.846, 2.54) and (1.128, 2.822) .. (1.128, 2.822);
    \draw[RedOrange, very thick] (1.128, 2.822) .. controls (1.285, 3.018) and (1.284, 3.212) .. (1.126, 3.406);
    \draw[RedOrange, very thick] (1.134, 3.396) -- (-0.327, 3.372);
    \node[disk normal] at (-1.981, 1.42) {};
    \node[disk normal] at (-1.963, 1.976) {};
    \node[disk normal] at (-2.044, 1.738) {};
    \node[disk normal] at (-0.558, 1.46) {};
    \node[disk normal] at (0.276, 1.965) {};
    \node[disk normal] at (1.11, 2.822) {};
    \node[disk normal] at (1.146, 3.381) {};
    \node[disk normal] at (-0.346, 3.372) {};
    \node[disk normal] at (-1.406, 2.815) {};
    \node[anchor=center] at (-2.268, 1.739) {$b_v$};
    \node[anchor=center] at (-0.642, 2.581) {$B$};
    \node[anchor=center] at (-0.307, 3.61) {$A_t$};
    \node[anchor=center] at (-1.871, 1.213) {$A_b$};
    \node[anchor=center] at (-2.101, 2.269) {$A_l$};
    \node[anchor=center] at (0.561, 1.92) {$A_r$};
    \node[anchor=center] at (1.225, 3.662) {$A_t$};
    \node[anchor=center] at (1.508, 2.721) {$A_r$};
    \node[anchor=center] at (-1.722, 2.969) {$A_l$};
    \node[anchor=center] at (-0.511, 1.193) {$A_b$};
    \node[anchor=center] at (-1.306, 1.75) {$B$};
    \node[anchor=center] at (0.708, 2.607) {$B$};
    \node[anchor=center] at (-0.003, 1.77) {$B$};
\end{tikzpicture}
\end{equation}

The map given by $b_v\to \Sq(v)$ is linear since the boundary $b_v + b_w$ results into $\Sq(v) + \Sq(w)$ for any system size due to the linearity of the trace. The point here is that given a PEPS with its boundary the maps $\sq^n$ can always be defined implicitly as
$$
\begin{tikzpicture}[scale=1]
    \draw[RedOrange, very thick] (-2.884, 1.289) .. controls (-2.259, 1.301) and (-2.258, 1.076) .. (-2.772, 0.726);
    \draw[RedOrange, very thick] (-2.772, 0.726) -- (-3.493, -0.098);
    \draw[RedOrange, very thick] (-4.555, -0.671) .. controls (-4.198, -0.756) and (-3.843, -0.567) .. (-3.489, -0.103);
    \draw[very thick] (1.005, 1.28) -- (-0.669, -0.697);
    \draw[very thick] (0.746, -0.665) -- (2.459, 1.314);
    \draw[very thick] (-0.079, 0.73) -- (2.461, 0.73);
    \draw[very thick] (-0.643, -0.116) -- (1.614, -0.116);
    \draw[RoyalBlue!80, very thick] (0.54, 0.73) -- (0.539, 1.072);
    \draw[RoyalBlue!80, very thick] (-0.178, -0.116) -- (-0.181, 0.246);
    \draw[RoyalBlue!80, very thick] (1.221, -0.116) -- (1.225, 0.237);
    \draw[RoyalBlue!80, very thick] (1.954, 0.73) -- (1.946, 1.058);
    \node[disk normal] at (0.54, 0.73) {};
    \node[disk normal] at (-0.178, -0.116) {};
    \node[disk normal] at (1.221, -0.116) {};
    \node[disk normal] at (1.954, 0.73) {};
    \draw[RedOrange, very thick] (1.005, 1.28) .. controls (0.587, 1.288) and (0.231, 1.104) .. (-0.062, 0.73);
    \draw[RedOrange, very thick] (-0.068, 0.73) -- (-0.643, -0.116);
    \draw[RedOrange, very thick] (-0.64, -0.111) .. controls (-0.734, -0.238) and (-0.736, -0.423) .. (-0.644, -0.667);
    \draw[RedOrange, very thick] (-0.644, -0.667) -- (0.753, -0.657);
    \draw[RedOrange, very thick] (0.753, -0.657) .. controls (1.117, -0.655) and (1.404, -0.475) .. (1.614, -0.116);
    \draw[RedOrange, very thick] (1.613, -0.116) .. controls (2.178, 0.448) and (2.46, 0.73) .. (2.46, 0.73);
    \draw[RedOrange, very thick] (2.46, 0.73) .. controls (2.617, 0.926) and (2.616, 1.12) .. (2.459, 1.314);
    \draw[RedOrange, very thick] (2.466, 1.304) -- (1.005, 1.28);
    \node[disk normal] at (-0.649, -0.672) {};
    \node[disk normal] at (-0.631, -0.116) {};
    \node[disk normal] at (-0.712, -0.354) {};
    \node[disk normal] at (0.774, -0.632) {};
    \node[disk normal] at (1.608, -0.127) {};
    \node[disk normal] at (2.442, 0.73) {};
    \node[disk normal] at (2.478, 1.289) {};
    \node[disk normal] at (0.986, 1.28) {};
    \node[disk normal] at (-0.073, 0.723) {};
    \node[anchor=center] at (-0.936, -0.353) {$b_v$};
    \node[anchor=center] at (0.69, 0.489) {$B$};
    \node[anchor=center] at (1.025, 1.518) {$A_t$};
    \node[anchor=center] at (-0.539, -0.879) {$A_b$};
    \node[anchor=center] at (-0.768, 0.177) {$A_l$};
    \node[anchor=center] at (1.893, -0.172) {$A_r$};
    \node[anchor=center] at (2.557, 1.57) {$A_t$};
    \node[anchor=center] at (2.84, 0.629) {$A_r$};
    \node[anchor=center] at (-0.39, 0.877) {$A_l$};
    \node[anchor=center] at (0.821, -0.899) {$A_b$};
    \node[anchor=center] at (0.026, -0.343) {$B$};
    \node[anchor=center] at (2.04, 0.515) {$B$};
    \node[anchor=center] at (1.329, -0.322) {$B$};
    \draw[very thick] (-2.913, 1.299) -- (-4.588, -0.677);
    \draw[RoyalBlue!80, very thick] (-3.379, 0.749) -- (-3.379, 1.092);
    \draw[RoyalBlue!80, very thick] (-4.096, -0.097) -- (-4.099, 0.265);
    \node[disk normal] at (-3.379, 0.749) {};
    \node[disk normal] at (-4.096, -0.097) {};
    \draw[RedOrange, very thick] (-2.913, 1.299) .. controls (-3.331, 1.307) and (-3.687, 1.124) .. (-3.981, 0.749);
    \draw[RedOrange, very thick] (-3.987, 0.749) -- (-4.562, -0.097);
    \draw[RedOrange, very thick] (-4.558, -0.092) .. controls (-4.652, -0.219) and (-4.654, -0.404) .. (-4.563, -0.648);
    \node[disk normal] at (-4.567, -0.653) {};
    \node[disk normal] at (-4.549, -0.097) {};
    \node[disk normal] at (-4.63, -0.335) {};
    \node[disk normal] at (-2.932, 1.299) {};
    \node[disk normal] at (-3.992, 0.742) {};
    \node[anchor=center] at (-4.854, -0.333) {$b_v$};
    \node[anchor=center] at (-3.228, 0.508) {$B$};
    \node[anchor=center] at (-4.457, -0.86) {$A_b$};
    \node[anchor=center] at (-4.687, 0.196) {$A_l$};
    \node[anchor=center] at (-4.309, 0.896) {$A_l$};
    \node[anchor=center] at (-3.892, -0.323) {$B$};
    \draw[very thick] (-3.975, 0.744) -- (-2.785, 0.738);
    \draw[very thick] (-4.565, -0.115) -- (-3.48, -0.094);
    \node[disk normal] at (-2.778, 0.737) {};
    \node[disk normal] at (-3.485, -0.106) {};
    \node[anchor=center] at (-2.834, 1.573) {$A_t$};
    \node[anchor=center] at (-2.486, 0.601) {$A_r$};
    \node[anchor=center] at (-3.22, -0.189) {$A_r$};
    \draw[very thick] (-4.983, 1.983) .. controls (-5.594, 0.963) and (-5.602, -0.078) .. (-5.008, -1.141);
    \draw[very thick] (-2.258, 1.977) .. controls (-1.87, 1.043) and (-1.861, 0.024) .. (-2.23, -1.08);
    \node[anchor=center] at (-5.808, 0.6) {$2$};
    \node[anchor=center] at (-5.792, 0.013) {$x$};
    \node[anchor=center, font=\large] at (-1.568, 0.514) {$=$};
    \draw[thick] (-6.38, 0.512) rectangle (-5.963, 0.133);
\end{tikzpicture} \ ,
$$
where the only requirement is to be injective, in particular for the previous case the map $ v \to \Sq^{2,1}(v)$, to prevent that $\Sq(v) = \Sq(w)$ if $v\neq w$.

\subsection{Two PEPS for the main example}
In this section we provide the PEPS tensor that realize the 2D coalgebra of Section \ref{sec:pivot}. We find two possible representations, one with smaller virtual level but a non-trivial MPS boundary and the other with bigger virtual dimension but trivial MPS boundary.

The smallest bond dimension PEPS tensor has $D=2$ and its non-zero components are:
\begin{equation}
     \begin{tikzpicture}[baseline=1pt]
    \node[tensor] at (0,0){};
    \draw[thick] (-0.5,0)--(0.5,0);
    \draw[thick] (0,-0.5)--(0,0.5);
    \draw[thick] (0,0)--(0.2,0.2);
    \end{tikzpicture} =
\setlength\arraycolsep{2pt}
  \begin{matrix}   &1&\\ 0 &v&1\\ & 0&\end{matrix} 
+ \begin{matrix} &1 & \\ 1 &b  &1 \\  & 0& \end{matrix} 
+ \begin{matrix} &1 & \\ 1 &b  &1 \\  & 1& \end{matrix} 
+ \begin{matrix} &0 & \\ 0 &a &0 \\  & 0& \end{matrix}
+ \begin{matrix} &1 & \\ 0 &a &0 \\  & 0& \end{matrix}
\end{equation}
This representation has an MPS boundary that correlates the right and left parts as a $W$-state with $\chi=2$ and $A_t= \bra{1}, A_b= \bra{0}$.

The other representation has $D=4$. The non-trivial components of the tensor are
\begin{align*}
         \begin{tikzpicture}[baseline=1pt]
    \node[tensor] at (0,0){};
    \draw[thick] (-0.5,0)--(0.5,0);
    \draw[thick] (0,-0.5)--(0,0.5);
    \draw[thick] (0,0)--(0.2,0.2);
    \end{tikzpicture} &= 
\setlength\arraycolsep{2pt}
 \begin{matrix} &2 & \\ 1 &b  &1 \\  & 2& \end{matrix} 
+ \begin{matrix} &3 & \\ 1 &b  &3 \\  & 3& \end{matrix}
+ \begin{matrix} &3 & \\ 3 &b  &3 \\  & 3& \end{matrix} 
\\ &
\setlength\arraycolsep{2pt}
+ \begin{matrix} &2 & \\ 0 &a &0 \\  & 0& \end{matrix}
+ \begin{matrix}   &3&\\ 0 &v&3\\ & 0&\end{matrix} 
+ \begin{matrix} &3 & \\ 3 &b  &3 \\  & 1& \end{matrix}
\\  & 
\setlength\arraycolsep{2pt}
+ \begin{matrix} &0 & \\ 0 &a &0 \\  & 0& \end{matrix}
+ \begin{matrix} &0 & \\ 0 &a &2 \\  & 0& \end{matrix}
+ \begin{matrix} &1 & \\ 2 &a &2 \\  &1& \end{matrix} \ ,
\end{align*}
and the MPS boundary is trivial, i.e. it is a tensor product state. The defining MPS tensors are $A_t= A_r= \bra{2}+\bra{3}$ and $A_b= A_l= \bra{0}+\bra{1}$.

\section{\texorpdfstring{Application to Quantum groups: $U_q[su(2)]$}{Application to Quantum groups: Uq[su(2)]}}
\label{secQG}

The quantum group $U_q[su(2)]$ is given by the basis $\{S^\pm, K^\pm \}$ which is a deformation of $su(2)$ where $K^\pm=q^{\pm S_z}$ \cite{Pinkbook}. The defining relations for the quantum group $U_q[su(2)]$ are:
\begin{align*}
    [S^+, S^-] &= \frac{K^{+2} - K^{-2}}{q - q^{-1}}, \\
    K^\alpha \cdot S^\pm &= q^{\pm \alpha} S^\pm \cdot K^\alpha, \quad \alpha = \pm 1.
\end{align*}

The coproduct for $U_q[su(2)]$ is defined as follows:
\begin{align}\label{defQGS}
    \Delta(S^\pm) &= S^\pm \otimes K^+ + K^- \otimes S^\pm, \\
    \Delta(K^\pm) &= K^\pm \otimes K^\pm. \notag
\end{align}
This structure generalizes to higher tensor powers:
\[
\Delta^n(S^\pm) = \sum_{p=0}^{n} (K^+)^{\otimes p} \otimes S^\pm \otimes (K^-)^{\otimes (n-p)}, \quad
\Delta^n(K^\pm) = (K^\pm)^{\otimes (n+1)}.
\]
The coproduct of the commutator $[S^+, S^-]$ also follows a similar pattern:
\[
\Delta^n([S^+, S^-]) = \sum_{p=0}^{n} (K^{+2})^{\otimes p} \otimes [S^+, S^-] \otimes (K^{-2})^{\otimes (n-p)} = \frac{\Delta^n( K^{+2}) - \Delta^n( K^{-2})}{q - q^{-1}}.
\]

\subsection{\texorpdfstring{The Quantum group $U_q[su(2)]$ in 2D}{The Quantum group Uq[su(2)] in 2D}}

As we have seen above, the coproduct of the quantum group $U_q[su(2)]$ satisfies the conditions needed to use the construction of the example of Section \ref{sec:pivot} to create a 2D coalgebra and its compatible algebra. In this section we will study this example in more detail. 
The 2D operators of $U_q[su(2)]$ that we propose have the following form:
\begin{align}\label{eqSqS}
\Sq^{n,m}(S^\pm)  &= \sum_{\substack{1\le i \le n \\ 1\le j \le m}} \hat{S}^\pm_{i,j}, \quad \hat{S}^\pm_{i,j} =
\setlength\arraycolsep{0.5pt}
         \begin{pmatrix}
           & & & K^+ & & &  \\
           &K^+ & & \vdots & & K^+ & \\
           & & & K^+ & & & \\
           K^-& \cdots & K^-& (S^\pm)_{(i,j)} &K^+ &\cdots &K^+ \\
           & & & K^- & & & \\
           &K^- & & \vdots & & K^-& \\
           & & & K^- & & &
         \end{pmatrix} \ , \\
         \label{eqSqK}
    \Sq^{n,m}(K^\pm) &=
             \begin{pmatrix}
           K^\pm&\cdots &K^\pm  \\
           \vdots& K^\pm& \vdots \\
          K^\pm & \cdots &K^\pm 
         \end{pmatrix} \equiv (K^\pm)^{\otimes nm} \ , \quad \Sq^{n,m}(S^z) = \sum_{i,j} S^z_{ij}\otimes \id  \ .
\end{align}

If $q\to 1$, these operators reduces to the non-deformed 2D case: $\Sq^{n,m}(S^\pm) = \sum S^\pm _{ij}\otimes \id$ and $\Sq^{n,m}(K^\pm)=\id$.

These operators can be constructed using the coproducts that have been defined in the section \ref{sec:pivot}. Moreover, as was shown there they have a compatible bialgebra structure. It is, however, instructive to write and check the algebraic relations explicitly. These operators reproduce the algebraic relations between $K^\pm$ and $S^\pm$ for any 2D system size: 
\begin{equation}
 \Sq(K^\alpha) \cdot \Sq(S^\pm) = q^{\pm \alpha} \Sq(S^\pm) \cdot \Sq(K^\alpha) \ , \ \alpha= \pm 1 \ .
 \end{equation}
The commutator of the 2D representation of $S^\pm$ satisfies the following:
\begin{align}\label{commutS}
    [\Sq(S^+),\Sq(S^-)] = & \ \sum_{i,j} [\hat{S}^+_{i,j},\hat{S}^-_{i,j}] =  \Sq([S^+,S^-])= \frac{\Sq(K^{+2})-\Sq(K^{-2})}{q-q^{-1}} \ .
    \end{align}
To show the previous equation note first that $[\Sq(S^+),\Sq(S^-)] = \sum_{i,j,i',j'} [\hat{S}^+_{ij},\hat{S}^-_{i'j'}] $. Let us then group the indices $\{i,j,i',j'\}$ into two disjoint sets: 
$\{i,j,i',j'\}= \{ i=i', j=j'\} \cup \{ (i,j)\neq (i',j')\} $. Then we can check that
$$ [\hat{S}^+_{ij},\hat{S}^-_{i'j'}] = 0 \ {\rm if} \ (i,j)\neq (i',j') \ ,$$
since it involves terms proportional to 
$$ S^+K^{\mp}\otimes K^{\pm}S^- - K^{\mp}S^+\otimes S^-K^{\pm}  = 0 \ ,$$
where the tensor product is over the sites $(i,j)$ and $(i',j')$. 
Therefore the only non-trivial contribution comes from the set $ \{ i=i', j=j'\}$, which results in 
$$\sum_{i,j} 
    \setlength\arraycolsep{1pt}
    \begin{pmatrix}
           & & & K^{+2} & & &  \\
           & K^{+2} & & \vdots & & K^{+2} & \\
           & & & K^{+2} & & & \\
           K^{-2}& \cdots & K^{-2}& [S^+,S^-]_{i,j} &K^{+2} &\cdots &K^{+2} \\
           & & & K^{-2} & & & \\
           &K^{-2} & & \vdots & & K^{-2}& \\
           & & & K^{-2} & & &
    \end{pmatrix}
    = \frac{1}{q-q^{-1}}
    \sum_{i,j} 
    \setlength\arraycolsep{1pt}
    \begin{pmatrix}
           & & & K^{+2} & & &  \\
           & K^{+2} & & \vdots & & K^{+2} & \\
           & & & K^{+2} & & & \\
           K^{-2}& \cdots & K^{-2}& (K^{+2}-K^{-2})_{i,j} &K^{+2} &\cdots &K^{+2} \\
           & & & K^{-2} & & & \\
           &K^{-2} & & \vdots & & K^{-2}& \\
           & & & K^{-2} & & &
    \end{pmatrix} \ .$$
The sum is telescopic thus we arrive at $(\bigotimes_{ij} K^{+2}_{ij} - \bigotimes_{ij} K^{-2}_{ij}) /(q-q^{-1})$, the desired result. For example,  for a $3\times 3$ lattice, the terms with the same color cancel each other:
\begin{align*}
    \begin{pmatrix}
           {\color{purple}K_+^2}{\color{red}-K_-^{2}}&K_+^2 &K_+^2  \\
           K_-^{2}& K_-^{2}& K_-^{2} \\
          K_-^{2} & K_-^{2} &K_-^{2} 
         \end{pmatrix} & +
    \begin{pmatrix}
          K_-^{2} &{\color{red}K_+^2}{\color{blue}-K_-^{2}}&K_+^2  \\
           K_-^{2}& K_-^{2}& K_-^{2} \\
          K_-^{2} & K_-^{2} &K_-^{2} 
         \end{pmatrix} +
    \begin{pmatrix}
          K_-^{2} &K_-^{2}&{\color{blue}K_+^2}-K_-^{2}  \\
           K_-^{2}& K_-^{2}& K_-^{2} \\
          K_-^{2} & K_-^{2} &K_-^{2} 
         \end{pmatrix} + \\
    \begin{pmatrix}
           K_+^{2}& K_+^{2}& K_+^{2} \\
           {\color{yellow}K_+^2}{\color{green}-K_-^{2}}&K_+^2 &K_+^2  \\
          K_-^{2} & K_-^{2} &K_-^{2} 
         \end{pmatrix} & +
    \begin{pmatrix}
           K_+^{2}& K_+^{2}& K_+^{2} \\
        K_-^{2} &{\color{green}K_+^2}{\color{orange}-K_-^{2}}&K_+^2  \\
          K_-^{2} & K_-^{2} &K_-^{2} 
         \end{pmatrix} +
    \begin{pmatrix}
           K_+^{2}& K_+^{2}& K_+^{2} \\
            K_-^{2} &K_-^{2}&{\color{orange}K_+^2}{\color{purple}-K_-^{2}}  \\
          K_-^{2} & K_-^{2} &K_-^{2} 
         \end{pmatrix} + \\
    \begin{pmatrix}
           K_+^{2}& K_+^{2}& K_+^{2} \\
          K_+^{2} & K_+^{2} &K_+^{2} \\
  K_+^2{\color{cyan}-K_-^{2}}&K_+^2 &K_+^2  
         \end{pmatrix} & +
    \begin{pmatrix}
           K_+^{2}& K_+^{2}& K_+^{2} \\
            K_+^{2} & K_+^{2} &K_+^{2} \\
        K_-^{2} &{\color{cyan}K_+^2}{\color{gray}-K_-^{2}}&K_+^2  
         \end{pmatrix} +
    \begin{pmatrix}
           K_+^{2}& K_+^{2}& K_+^{2}  \\
          K_+^{2} & K_+^{2} &K_+^{2} \\
        K_-^{2} &K_-^{2}&{\color{gray}K_+^2}{\color{yellow}-K_-^{2} }
         \end{pmatrix} \ ,
\end{align*}
and the black ones are the remaining ones.
For completeness let us write down the definition of the counits:
\begin{align*}
     \epsilon^n \left((K^-)^{\otimes p }\otimes S^\pm \otimes (K^+)^{\otimes n-p-1 } \right) &  = 0 \ , \ 0\le p < n \\
     \epsilon_y^n\left ((K^\pm )^{\otimes n}\right )& = 1\\
     \epsilon_x^n\left( ( K^+ )^{\otimes n} \otimes (K^- )^{\otimes n-p} \right)^T & = 1\\
\end{align*}
and the compatible antipode maps:
\begin{align*}
     S^n \left((K^+)^{\otimes p }\otimes S^\pm \otimes (K^-)^{\otimes n-p-1 } \right ) & = -q^\pm \left((K^+)^{\otimes p }\otimes S^\pm \otimes (K^-)^{\otimes n-p-1 } \right ) \, \ 0\le p < n \\
     S^n \left(1^{\otimes p }\otimes S^z \otimes 1^{\otimes n-p-1 }\right ) &= -\left(1^{\otimes p }\otimes S^z \otimes 1^{\otimes n-p-1 }\right ) \ , \ 0\le p < n \\
     S^n\left( (K^\pm)^{ \otimes n} \right ) & = (K^\mp)^{ \otimes n}\ .
\end{align*}
For the generating elements of the algebra $\{S^\pm, K^\pm\}$, the commutative diagrams that come from the $xy$-compatibility condition look like:
$$
\begin{tikzcd}
S^\pm\arrow{r}{\sq^1_x} \arrow{d}{\sq^1_y} & S^\pm\ K^+ + K^-\ S^\pm \arrow{d}{\sq^2_y} \\
\setlength\arraycolsep{0.5pt}
  \begin{matrix}  S^\pm \\ K^-\end{matrix} + 
  \begin{matrix}  K^+\\ S^\pm\end{matrix} \arrow{r}{\sq^2_x}  & 
\setlength\arraycolsep{0.5pt}
  \begin{matrix}  S^\pm&K^+\\ K^- &K^-\end{matrix} + \begin{matrix}  K^-&S^\pm\\ K^- &K^-\end{matrix} +
  \begin{matrix}  K^+&K^+\\ S^\pm &K^+\end{matrix}+
  \begin{matrix}  K^+&K^+\\ K^- &S^\pm\end{matrix}
\end{tikzcd}
\ , \quad 
\begin{tikzcd}
K^\pm\arrow{r}{\sq^1_x} \arrow{d}{\sq^1_y} & K^\pm \ K^\pm \arrow{d}{\sq^2_y} \\
\setlength\arraycolsep{0.5pt}
  \begin{matrix}  K^\pm \\ K^\pm\end{matrix} \arrow{r}{\sq^2_x}  & 
\setlength\arraycolsep{0.5pt}
  \begin{matrix}  K^\pm &K^\pm\\ K^\pm &K^\pm \end{matrix}
\end{tikzcd}
$$

\subsection{Invariant states: q-singlets}

We focus on a $2\times 2$ spin $1/2$ lattice ordered as $\begin{pmatrix}  1&2\\ 3 &4\end{pmatrix}$.
The $q$-deformed singlet on sites $i,j$ is defined as
$$\ket{s}^q_{i,j} = \frac{1}{\sqrt{q-q^{-1}}}(q^{1/2}\ket{01}_{i,j}-q^{-1/2}\ket{10}_{i,j})$$
and the algebra acts on it as:
\begin{align*}
\Delta(K^\pm) \ket{s}^q &=\ket{s}^q\ , \\
\Delta(S^\pm) \ket{s}^q &=0 \ .
\end{align*}

For our work the following relation is important:
$$(K^-\otimes K^+)\ket{s}^q = \frac{\sqrt{q^{-1}-q}}{\sqrt{q-q^{-1}}}\ket{s}^{q^{-1}} \ . $$

Let us rewrite the operators $S^\pm$ acting on this lattice in two equivalent ways:

\begin{align}
\Sq^{2,2}(S^\pm) &=
\setlength\arraycolsep{0.5pt}
  \begin{pmatrix}  S^\pm&K^+\\ K^- &K^-\end{pmatrix} + 
  \begin{pmatrix}  K^-&S^\pm\\ K^- &K^-\end{pmatrix} +
  \begin{pmatrix}  K^+&K^+\\ S^\pm &K^+\end{pmatrix}+
  \begin{pmatrix}  K^+&K^+\\ K^- &S^\pm\end{pmatrix} \\
  & = 
  \Delta(S^\pm)_{1,2}\otimes \Delta(K^-)_{3,4}+ \Delta(K^+)_{1,2}\otimes \Delta(S^\pm)_{3,4} \\
  & =
  K^-_3 K^+_2\otimes \Delta(S^\pm)_{4,1} 
  + 
 S^\pm_{3}K^+_2\otimes\Delta(K^+)_{1,4} +
  K^-_{3} S^\pm_2\otimes\Delta(K^-)_{1,4} \ .
\end{align}
For four spins $1/2$, its tensor product decomposed into the the follosing direct sum of $su(2)$ irreps:
$$\left (\frac{1}{2}\right)^{\otimes 4} = 0(\times2)\oplus 1(\times 3)\oplus 2 \ ,$$ 
where $(\times n)$ stands for a multiplicity $n$. It turns out that the two $s=0$ spaces are given by any combination of two horizontal $q$-singlets and two crossed $q$-singlets since they are annihilated by the plaquette operator $S^+$:

$$ \Sq^{2,2}(S^\pm) \cdot \left(  \alpha \ket{s}^q_{1,2}   \ket{s}^q_{3,4} + \beta \ket{s}^q_{3,2}  \ket{s}^q_{4,1}\right) = 0, \  \forall \alpha,\beta \in \mathbb{C} \ .$$
This kernel space can be depicted in the lattice as:
$$
\begin{tikzpicture}[scale=1]
    \node[disk large] at (-3.669, 1.129) {};
    \node[disk large] at (-2.54, 1.129) {};
    \node[disk large] at (-3.669, 2.258) {};
    \node[disk large] at (-2.54, 2.258) {};
    \node[disk large] at (-0.854, 1.125) {};
    \node[disk large] at (0.275, 2.254) {};
    \draw[very thick] (-3.669, 1.129) -- (-2.54, 1.129) -- (-2.54, 1.129);
    \draw[very thick] (-3.669, 2.258) -- (-2.54, 2.258);
    \draw[very thick] (-0.854, 1.125) -- (0.275, 2.254);
    \draw[very thick, thick bevel] (-0.854, 2.254) -- (0.275, 1.125);
    \draw (-3.951, 2.54) -- (-3.951, 0.847) -- (-3.951, 0.847);
    \draw (-1.136, 2.536) -- (-1.136, 0.843);
    \draw (-2.383, 2.553) -- (-2.1, 1.706) -- (-2.383, 0.859);
    \draw (0.386, 2.541) -- (0.669, 1.694) -- (0.386, 0.848);
    \node[disk large] at (-0.854, 1.125) {};
    \node[disk large] at (0.275, 2.254) {};
    \draw[very thick] (-0.854, 1.125) -- (0.275, 2.254);
    \draw[very thick, thick bevel] (-0.854, 2.254) -- (0.275, 1.125);
    \draw (-1.136, 2.536) -- (-1.136, 0.843);
    \draw (0.386, 2.541) -- (0.669, 1.694) -- (0.386, 0.848);
    \node[anchor=center] at (-1.693, 1.693) {$+ \beta$};
    \node[anchor=center] at (-4.233, 1.693) {$\alpha$};
    \node[disk large] at (-0.854, 2.254) {};
    \node[disk large] at (-0.854, 2.254) {};
    \node[disk large] at (0.275, 1.125) {};
    \node[disk large] at (0.275, 1.125) {};
\end{tikzpicture} \ .
$$
However, two vertical $q$-singlets are not in the kernel of this operator:
$$ \Sq^{2,2}(S^\pm) \cdot \ket{s}^q_{1,3}   \ket{s}^q_{2,4} = \frac{q^{1/2}-q^{-1/2}}{\sqrt{q-q^{-1}}}  (\ket{01}+\ket{10})_{1,3}\ket{11}_{2,4} + \ket{11}_{1,3}(\ket{01}+\ket{10})_{2,4}  \neq 0 \ ,$$
but in the limit $q\to 1$, the coefficient goes to zero recovering the classical result.

\subsection{The R-matrix in 2D}

One of the central structures in the theory of quantum groups is the R-matrix, which encodes the nontrivial braiding of representations and underpins the integrability of quantum systems. For the quantum group $U_q[su(2)]$, a q-deformation of the universal enveloping algebra of $su(2)$, the R-matrix plays a pivotal role in defining the quasitriangular Hopf algebra structure. It provides a solution to the quantum Yang–Baxter equation, ensuring consistency of particle-like excitations in low-dimensional quantum systems, and enables the construction of braid group representations, which are essential in applications to knot theory and topological quantum computation. In this section we first review the connstructionn in the 1D case and propose the solution for our 2D case.

\subsubsection{The one-dimensional case}
We first explain the one-dimensional $R$-matrix. The $R$-matrix is the intertwiner of \eqref{defQGS} with the permuted comultiplication $ \Delta^{per} $ defined as 
\beq
 \Delta^{per}  ( S^\pm)  = S^\pm  \otimes K^{-1}  + K \otimes S^\pm  
 \label{3}\ ,
 \eeq
that is,
 \beq
 R(q) \Delta (S^\pm)  = \Delta^{per} (S^\pm) R(q)
 \label{4} \ .
 \eeq
For the spin 1/2 representation one has
\beq
R(q) = 
\left( \begin{array}{cccc}
q & 0 & 0 & 0 \\
0 & 1 & q - q^{-1}  & 0 \\
0 & 0 & 1 & 0 \\
0 & 0 & 0 & q \\
\end{array}
\right)
= q^{ 2 S^z \otimes S^z} q^{\frac{1}{2}}  \left(  \mathbb{I} + (q - q^{-1} ) S^+ \otimes S^- \right)
\label{5}
\eeq
Let us write
\beq
q = e^h
\label{6} \ ,
\eeq
and take the semiclassical limit $h \ll 1$:
\beq
R(q) = \mathbb{I} + 2 h \, { \bf r} + O(h^2)
\label{7} \ ,
\eeq
where ${\bf r}$ denotes the classical $r$-matrix that is given by
\beq
{\bf r} = 
\left( \begin{array}{cccc}
\frac{1}{2} & 0 & 0 & 0 \\
0 & 0 & 1  & 0 \\
0 & 0 & 0 & 0 \\
0 & 0 & 0 & \frac{1}{2}  \\
\end{array}
\right)
= \frac{1}{4} + S^z_1 S^z_2 + S^+_1 S^-_2
\label{8} \ .
\eeq
The limit of the 1D-coproducts are
\barray
 \Delta S^\pm & = & S^\pm_1  +  S^\pm_2 + h ( S_1^\pm S^z_2 - S^z_1 S^\pm _2) + O(h^2) \label{9}  \ , \\
  \Delta^{per}  S^\pm & = & S^\pm_1 +  S^\pm_2 -  h ( S_1^\pm S^z_2 - S^z_1 S^\pm _2) + O(h^2)  \label{10} \ .
 \earray 
Eq. \eqref{4} becomes  
\beq
[r_{12},  S^\pm_1 + S^\pm_2 ] = -   S_1^\pm S^z_2 + S^z_1 S^\pm _2 + O(h) \ .
\label{11}
\eeq
As can be checked using the previous equations. 
 
\subsubsection{The proper 2D case}
  
The permuted comultiplication is defined with the replacement $K^+ \leftrightarrow K^{-1}$:
  
 \barray 
\widetilde{\Sq}^{2,2}(S^\pm) &= 
\setlength\arraycolsep{0.5pt}
  \begin{pmatrix}  S^\pm&K^-\\ K^+ &K^+\end{pmatrix} + 
  \begin{pmatrix}  K^+&S^\pm\\ K^+ &K^+\end{pmatrix} +
  \begin{pmatrix}  K^-&K^-\\ S^\pm &K^-\end{pmatrix}+
  \begin{pmatrix}  K^-&K^-\\ K^+ &S^\pm\end{pmatrix}    \label{12} \\
  & = 
  \Delta^{per}(S^\pm)_{1,2}\otimes \Delta(K^+)_{3,4}+ \Delta(K^-)_{1,2}\otimes \Delta^{per}(S^\pm)_{3,4}  \nonumber 
\earray 
We want  to find the operator $\mathfrak{R}(q)$ that satisfies 
 \beq
 \mathfrak{R}(q)\Sq^{2,2}(S^\pm)   = 
  \widetilde{\Sq}^{2,2}(S^\pm)  \mathfrak{R}(q)
 \label{13} \ .
 \eeq 
As in the 1D case we consider the semiclassical limit of Eq.~\eqref{7}:
 \barray 
\Sq^{2,2}(S^\pm) &=  & S^\pm_1  +  S^\pm_2 +  S^\pm_3  +  S^\pm_4 +   h  S_1^\pm ( S^z_2 - S^z_3 - S^z_4) +  h  S_2^\pm ( - S^z_1 - S^z_3 - S^z_4) 
 \label{14}  \\
& + &     h  S_3^\pm ( S^z_1 + S^z_2 + S^z_4)   +   h  S_4^\pm ( S^z_1 + S^z_2 - S^z_3)        + O(h^2) \ .
\nonumber 
\earray 
Therefore Eq.~\eqref{13} in the first order on $h$ reads as
\barray
[ \mathfrak{r}_{1234},   S^\pm_1  +  S^\pm_2 +  S^\pm_3  +  S^\pm_4 ] &  =  &   S_1^\pm (- S^z_2 + S^z_3 + S^z_4) +    S_2^\pm (  S^z_1 + S^z_3 + S^z_4) 
\label{15} \\
& - &   S_3^\pm ( S^z_1 + S^z_2 + S^z_4)   -     S_4^\pm ( S^z_1 + S^z_2 - S^z_3)  \ .\nonumber 
\earray 
To solve this equation let us define
\beq
A^\pm_{ij} =  -   S_i^\pm S^z_j + S^z_i S^\pm _j\ , \  i, j =1, \dots, 4 \ .
\label{16}
\eeq
According to Eq.~\eqref{11}, $A^\pm_{ij}$ satisfies
\beq
[r_{ij},  S^\pm_i + S^\pm_j ] = A^\pm_{ij} 
\label{17} \ .
\eeq
Using Eq.~\eqref{16},  Eq.~\eqref{15} can be written as 
\barray
[ \mathfrak{r}_{1234},   S^\pm_1  +  S^\pm_2 +  S^\pm_3  +  S^\pm_4 ] &  =  &   A^\pm_{12} +  A^\pm_{34} -  A^\pm_{14} -  A^\pm_{13} -  A^\pm_{23}-  A^\pm_{24} 
  \label{18} \ ,
\earray
that can be solved by 
\barray
 \mathfrak{r}_{1234} &  =  &   r_{12} + r_{34} -  r_{14} -r_{13} -  r_{23}-  r_{24} \ .
  \label{19} 
\earray
This solution is unique up to an additive constant. 
The previous equations indicate that the matrix $ \mathfrak{R}$ should be a product of $R$ matrices following a certain order
related to the signs of the classical $r$ matrices in Eq.~\eqref{19}. 
To show that this is the case, let us write Eq.~\eqref{3} and Eq.~\eqref{4}, and its inverse as follows,
\barray 
R_{12} ( S^\pm_1  K^{+}_2  + K_1^-  S^\pm_2 )R_{12}^{-1}  & = & S^\pm_1  K^{-}_2  + K^+_2  S^\pm_2 
\label{22} \ , \\
R_{12}^{-1}  ( S^\pm_1  K^{-}_2  + K_1^+  S^\pm_2 )R_{12}  & = & S^\pm_1  K^{+}_2  + K^-_2  S^\pm_2 \ .
\label{23}
\earray 
We shall now act with $R_{ij}\left( \bullet \right) R^{-1}_{ij}$ on $\Sq^{2,2}(S^\pm)$ several times using the previous equations as follows:
  \barray 
\Sq^{2,2}(S^\pm) &= 
\setlength\arraycolsep{0.5pt}
  \begin{pmatrix}  S^\pm&K^+\\ K^- &K^-\end{pmatrix} + 
  \begin{pmatrix}  K^-&S^\pm\\ K^- &K^-\end{pmatrix} +
  \begin{pmatrix}  K^+&K^+\\ S^\pm &K^+\end{pmatrix}+
  \begin{pmatrix}  K^+&K^+\\ K^- &S^\pm\end{pmatrix}   \\
  \stackrel{R_{12} \left( \bullet \right) R_{12}^{-1}}{\Longrightarrow}  & =
  \setlength\arraycolsep{0.5pt}
  \begin{pmatrix}  S^\pm&K^-\\ K^- &K^-\end{pmatrix} + 
  \begin{pmatrix}  K^+&S^\pm\\ K^- &K^-\end{pmatrix} +
  \begin{pmatrix}  K^+&K^+\\ S^\pm &K^+\end{pmatrix}+
  \begin{pmatrix}  K^+&K^+\\ K^- &S^\pm\end{pmatrix} 
  \nonumber   \\
  \stackrel{R_{34} \left( \bullet \right) R_{34}^{-1}}{\Longrightarrow}  & =
  \setlength\arraycolsep{0.5pt}
  \begin{pmatrix}  S^\pm&K^-\\ K^- &K^-\end{pmatrix} + 
  \begin{pmatrix}  K^+&S^\pm\\ K^- &K^-\end{pmatrix} +
  \begin{pmatrix}  K^+&K^+\\ S^\pm &K^-\end{pmatrix}+
  \begin{pmatrix}  K^+&K^+\\ K^+ &S^\pm\end{pmatrix} 
  \nonumber  
   \\
  \stackrel{R_{23}^{-1} \left( \bullet \right) R_{23}}{\Longrightarrow}  & =
  \setlength\arraycolsep{0.5pt}
  \begin{pmatrix}  S^\pm&K^-\\ K^- &K^-\end{pmatrix} + 
  \begin{pmatrix}  K^+&S^\pm\\ K^+ &K^-\end{pmatrix} +
  \begin{pmatrix}  K^+&K^-\\ S^\pm &K^-\end{pmatrix}+
  \begin{pmatrix}  K^+&K^+\\ K^+ &S^\pm\end{pmatrix} 
  \nonumber  
  \\
  \stackrel{R_{13}^{-1} \left( \bullet \right) R_{13}}{\Longrightarrow}  & =
  \setlength\arraycolsep{0.5pt}
  \begin{pmatrix}  S^\pm&K^-\\ K^+ &K^-\end{pmatrix} + 
  \begin{pmatrix}  K^+&S^\pm\\ K^+ &K^-\end{pmatrix} +
  \begin{pmatrix}  K^-&K^-\\ S^\pm &K^-\end{pmatrix}+
  \begin{pmatrix}  K^+&K^+\\ K^+ &S^\pm\end{pmatrix} 
  \nonumber  
   \\
  \stackrel{R_{24}^{-1} \left( \bullet \right) R_{24}}{\Longrightarrow}  & =
  \setlength\arraycolsep{0.5pt}
  \begin{pmatrix}  S^\pm&K^-\\ K^+ &K^-\end{pmatrix} + 
  \begin{pmatrix}  K^+&S^\pm\\ K^+ &K^+\end{pmatrix} +
  \begin{pmatrix}  K^-&K^-\\ S^\pm &K^-\end{pmatrix}+
  \begin{pmatrix}  K^+&K^-\\ K^+ &S^\pm\end{pmatrix} 
  \nonumber 
    \\
  \stackrel{R_{14}^{-1} \left( \bullet \right) R_{14}}{\Longrightarrow}  & =
  \setlength\arraycolsep{0.5pt}
  \begin{pmatrix}  S^\pm&K^-\\ K^+ &K^+\end{pmatrix} + 
  \begin{pmatrix}  K^+&S^\pm\\ K^+ &K^+\end{pmatrix} +
  \begin{pmatrix}  K^-&K^-\\ S^\pm &K^-\end{pmatrix}+
  \begin{pmatrix}  K^-&K^-\\ K^+ &S^\pm\end{pmatrix} 
  \nonumber  \\
  & =  \widetilde{\Sq}^{2,2}(S^\pm) & \nonumber 
  \earray 
  which yields the 2D $\mathfrak{R}$ matrix in terms of products of 1D $R$ matrices
  \beq
\mathfrak{R}_{1234} =  R^{-1}_{14} R_{24}^{-1} R_{13}^{-1} R_{23}^{-1} R_{34} R_{12} \ .
 \label{25}
 \eeq
 The semiclassical expansion of this equation agrees with Eq.~\eqref{19}.

\section{Conclusion and outlook}

In this paper we have explored a possible generalization of coproducts to the two dimensional square lattice together with their Hopf-algebra structures. We have focused on the quantum groups $U_q[su(2)]$, but the same construction is valid for the quantum group generated by $\{E,F,K,K^{-1}\}$ with $\Delta(E)= E\otimes K + 1\otimes E$, $\Delta(F)= F\otimes 1 + K^{-1}\otimes F$ and $[E,F]=(K-K^{-1}) /(q-q^{-1})$.

A central object in the formalism is the 2D R-matrix that intertwines the 2D coproducts of the coalgebraic symmetry. It will be very interesting to explore whether this R-matrix satisfies generalized Yang–Baxter equations, along the lines of the recent proposal by Korepin et al. in Ref.~\cite{Padmanabhan_2024}, where higher-dimensional or categorified Yang–Baxter structures are introduced. Understanding such generalized braid relations could deepen the connection between algebraic symmetry and integrability in 2D lattice systems. The existence of a well-defined R-matrix naturally leads to the study of topological invariants of knots and links that arise from the representation theory of the associated bialgebra. It would be compelling to investigate which invariants can be extracted from the 2D R-matrix in this setting, and whether they form a subclass or deformation of known invariants such as the Jones polynomial. Since the Jones polynomial arises from the standard quantum group $U_q[su(2)]$, any deviation or restriction induced by the 2D lattice structure may provide novel insights into categorified topological quantum field theories (TQFTs) and their lattice realizations.

A very interesting open problem is the study of 2D local Hamiltonians commuting with the operators of the 2D quantum group $U_q[su(2)]$. Moreover, we wonder what is the connection between the representation theory of these 2D Hopf algebras with fusion 2-categories. We leave the representation of these coalgebras in terms of PEPS for future work.

\section*{Acknowledgments}

JGR is funded by the FWF Erwin Schrödinger Program (Grant DOI 10.55776/J4796).
AM acknowledges support by the Austrian Science Fund(FWF) via Grants 10.55776/COE1, by the European Union – NextGenerationEU, and by the European Union’s Horizon 2020 research and innovation programme through Grant No. 863476". G.S. acknowledges financial support from the Spanish MINECO grant PID2021-127726NB-I00, the CSIC Research Platform on Quantum Technologies PTI-001, the QUANTUM ENIA project Quantum Spain funded through the RTRP-Next Generation program under the framework of
the Digital Spain 2026 Agenda and partial support from NSF grant PHY-2309135 to the Kavli Institute
for Theoretical Physics (KITP), as well as joint sponsorship from the Fulbright Program and the Spanish
Ministry of Science, Innovation and Universities.

\bibliographystyle{alpha}
\bibliography{bibliography}

\newcommand{\etalchar}[1]{$^{#1}$}
\begin{thebibliography}{MdAGR{\etalchar{+}}22}

\bibitem[AM06]{andruskiewitsch2006examples}
Nicol{\'a}s Andruskiewitsch and Juan~Mart{\'i}n Mombelli.
\newblock Examples of weak hopf algebras arising from vacant double groupoids.
\newblock {\em Nagoya Mathematical Journal}, 181:1--27, 2006.

\bibitem[BMW{\etalchar{+}}17]{Bultinck17A}
N.~Bultinck, M.~Mariën, D.J. Williamson, M.B. Şahinoğlu, J.~Haegeman, and F.~Verstraete.
\newblock Anyons and matrix product operator algebras.
\newblock {\em Annals of Physics}, 378:183--233, 2017.

\bibitem[BPSN{\etalchar{+}}24]{bhardwaj2024gappedphases21dnoninvertible}
Lakshya Bhardwaj, Daniel Pajer, Sakura Schafer-Nameki, Apoorv Tiwari, Alison Warman, and Jingxiang Wu.
\newblock Gapped phases in (2+1)d with non-invertible symmetries: Part i, 2024.

\bibitem[BS99]{bohm1999weak}
Gabriella B{\"o}hm and Korn{\'e}l Szlach{\'a}nyi.
\newblock Weak hopf algebras ii. representation theory, dimensions, and the markov trace.
\newblock {\em arXiv preprint math/9906045}, 1999.

\bibitem[BSNTW25]{bhardwaj2025gappedphases21dnoninvertible}
Lakshya Bhardwaj, Sakura Schafer-Nameki, Apoorv Tiwari, and Alison Warman.
\newblock Gapped phases in (2+1)d with non-invertible symmetries: Part ii, 2025.

\bibitem[Dri86]{DrinfeldQG}
V.G. Drinfeld.
\newblock Quantum groups.
\newblock {\em Proceedings International Congress of Mathematicians}, 1986.

\bibitem[DT24]{Delcamp_2024}
Clement Delcamp and Apoorv Tiwari.
\newblock Higher categorical symmetries and gauging in two-dimensional spin systems.
\newblock {\em SciPost Physics}, 16(4), April 2024.

\bibitem[ES24]{etingof2024briefintroductionquantumgroups}
Pavel Etingof and Mykola Semenyakin.
\newblock A brief introduction to quantum groups, 2024.

\bibitem[GKSW15]{Gaiotto_2015}
Davide Gaiotto, Anton Kapustin, Nathan Seiberg, and Brian Willett.
\newblock Generalized global symmetries.
\newblock {\em Journal of High Energy Physics}, 2015(2), February 2015.

\bibitem[GRAS96]{Pinkbook}
Cesar Gómez, Mart\'i Ruiz-Altaba, and German Sierra.
\newblock {\em Quantum Groups in Two-Dimensional Physics}.
\newblock Cambridge Monographs on Mathematical Physics. Cambridge University Press, 1996.

\bibitem[GRLM23]{Garre22_MPOSYM}
José Garre-Rubio, Laurens Lootens, and András Molnár.
\newblock Classifying phases protected by matrix product operator symmetries using matrix product states.
\newblock {\em Quantum}, 7:927, February 2023.

\bibitem[GRM25]{garreTN2DSYM}
José Garre-Rubio and András Molnár.
\newblock On two-dimensional tensor network group symmetries, 2025.

\bibitem[IHTSN25]{inamura202521dlatticemodelstensor}
Kansei Inamura, Sheng-Jie Huang, Apoorv Tiwari, and Sakura Schafer-Nameki.
\newblock (2+1)d lattice models and tensor networks for gapped phases with categorical symmetry, 2025.

\bibitem[IO24]{Inamura_2024}
Kansei Inamura and Kantaro Ohmori.
\newblock Fusion surface models: 2+1d lattice models from fusion 2-categories.
\newblock {\em SciPost Physics}, 16(6), June 2024.

\bibitem[MdAGR{\etalchar{+}}22]{molnar22}
Andras Molnar, Alberto~Ruiz de~Alarcón, José Garre-Rubio, Norbert Schuch, J.~Ignacio Cirac, and David Pérez-García.
\newblock Matrix product operator algebras i: representations of weak hopf algebras and projected entangled pair states, 2022.

\bibitem[Nas]{nastase_yangbaxter}
Horatiu Nastase.
\newblock Yang-baxter equation, bethe ansatz, yangian, lax pairs, toda and calogero-moser systems.

\bibitem[Oru13]{orus2013practical}
Roman Orus.
\newblock A practical introduction to tensor networks: Matrix product states and projected entangled pair states.
\newblock {\em arXiv preprint arXiv:1306.2164}, 2013.

\bibitem[PK24]{Padmanabhan_2024}
Pramod Padmanabhan and Vladimir Korepin.
\newblock Solving the yang-baxter, tetrahedron and higher simplex equations using clifford algebras.
\newblock {\em Nuclear Physics B}, 1007:116664, October 2024.

\bibitem[PS90]{PASQUIER1990523}
V.~Pasquier and H.~Saleur.
\newblock Common structures between finite systems and conformal field theories through quantum groups.
\newblock {\em Nuclear Physics B}, 330(2):523--556, 1990.

\bibitem[Que20]{QuellaQGSPT}
Thomas Quella.
\newblock Symmetry-protected topological phases beyond groups: The $q$-deformed affleck-kennedy-lieb-tasaki model.
\newblock {\em Phys. Rev. B}, 102:081120, Aug 2020.

\bibitem[Ros]{roscilde_xxz}
Tommaso Roscilde.
\newblock Td1: Exact diagonalization of the s = 1/2 xxz spin chain.

\bibitem[RT19]{Thorngren19}
Yifan~Wang Ryan~Thorngren.
\newblock Fusion category symmetry i: Anomaly in-flow and gapped phases.
\newblock {\em arXiv:1912.02817}, 2019.

\bibitem[Taf71]{TAFT}
Earl~J. Taft.
\newblock The order of the antipode of finite-dimensional hopf algebra.
\newblock {\em Proceedings of the National Academy of Sciences}, 68(11):2631--2633, 1971.

\bibitem[VC05]{Verstraete04}
F.~Verstraete and J.~I. Cirac.
\newblock Renormalization algorithms for quantum-many body systems in two and higher dimensions.
\newblock {\em ArXiv: cond-mat/0407066}, 2005.

\end{thebibliography}

\end{document}